\documentclass[12pt,draftcls,onecolumn]{IEEEtran}
\pdfoutput=1
\usepackage{mathtools} 
\usepackage{latexsym}
\usepackage{graphicx}
\usepackage{array}
\usepackage{amsmath}
\usepackage{amsfonts}
\usepackage{amssymb}
\usepackage{amsthm}
\usepackage{epsfig}
\usepackage{cite}
\usepackage{tabulary}
\usepackage{booktabs}
\renewcommand{\thefigure}{\thesection.\arabic{figure}}
\renewcommand{\theequation}{\thesection.\arabic{equation}}

\newtheorem{theorem}{Theorem}[section]
\newtheorem{lemma}[theorem]{Lemma}
\newtheorem{definition}[theorem]{Definition}

\newtheorem{remark}[theorem]{Remark}

\newtheorem{assumption}[theorem]{Assumption}

\newcommand{\sr}{\stackrel}

\newcommand{\tri}{\sr{\triangle}{=}}

\newcommand{\be}{\begin{equation}}
\newcommand{\ee}{\end{equation}}
\newcommand{\bea}{\begin{eqnarray}}
\newcommand{\eea}{\end{eqnarray}}
\newcommand{\bes}{\begin{eqnarray*}}
\newcommand{\ees}{\end{eqnarray*}}
\newcommand{\bi}{\begin{itemize}}
\newcommand{\ei}{\end{itemize}}
\newcommand{\ben}{\begin{enumerate}}
\newcommand{\een}{\end{enumerate}}
\newcommand{\bp}{\begin{problem}}
\newcommand{\ep}{\end{problem}}

\newcommand{\noi}{\noindent}

\renewcommand{\IEEEQED}{\IEEEQEDclosed}
\begin{document}
%\bibliographystyle{ieeetr}
%\baselineskip=18pt

%\begin{flushright}
%\tiny{\today}
%\end{flushright}

\title{Nonanticipative Rate Distortion Function and Relations to Filtering Theory}
%
%\vspace*{1.0cm}
%

\author{Charalambos~D.~Charalambous, Photios~A.~Stavrou,~\IEEEmembership{Student Member,~IEEE,}
        and~Nasir~U.~Ahmed% <-this % stops a space
\thanks{This work was financially supported by a medium size University of Cyprus grant entitled ``DIMITRIS" and by European Community's Seventh Framework Programme (FP7/2007-2013) under grant agreement no. INFSO-ICT-223844. Part of this work was presented in 20$^{th}$ International Symposium on Mathematical Theory of Networks and Systems (MTNS '12)\cite{stavrou-charalambous2012b}.}
\thanks{Charalambos D. Charalambous and Photios A. Stavrou are with the Department of Electrical and Computer Engineering (ECE), University of Cyprus, 75 Kallipoleos Avenue, P.O. Box 20537, Nicosia, 1678, Cyprus, e-mail: chadcha@ucy.ac.cy, stavrou.fotios@ucy.ac.cy}
\thanks{Nasir U. Ahmed is with the School of Electrical Engineering and Computer Science,
University of Ottawa, 161 Louis Pasteur, P.O. Box 450, Stn A, Ottawa, Ontario, Canada, K1N 6N5, e-mail: ahmed@site.uottawa.ca}}
%}
%
%\vspace*{0.5cm}

%\date{}

\maketitle
%\thispagestyle{empty}
%\vspace*{1.cm}

\begin{abstract}
The relation between nonanticipative Rate Distortion Function (RDF) and filtering theory is discussed on abstract spaces. The relation is established by imposing a realizability constraint on the reconstruction conditional distribution of the classical RDF. Existence of the extremum solution of the nonanticipative RDF is shown using weak$^*$-convergence on appropriate topology. The extremum  reconstruction conditional distribution is derived in closed form, for the case of stationary processes. The realization of the reconstruction conditional distribution which achieves the infimum of the nonanticipative RDF is described. Finally, an example is presented to illustrate the concepts.
\end{abstract}

%\noi{\bf AMS subject classification} 28A33, 46E15, 60G07

\section{Introduction}\label{introduction}
Shannon's information theory \cite{shannon1948} for reliable communication evolved over the years without much emphasis on nonanticipation imposed on the communication sub-systems. 
%However, recent interest in reliable communication for filtering and control applications \cite{yuksel2011,gupta-dana-hespanha-murray-hassibi2009,freudenberg-middleton2008,borkar-mitter-tatikonda2001,tatikonda-mitter2004,yuksel-meyn2012} raises the question whether some of the information theoretic measures can be modified to account for the nonanticipative processing of information.
In particular, the classical rate distortion function (RDF) for source data compression or quantization deals with the characterization of the optimal reconstruction conditional  distribution subject to a fidelity criterion \cite{berger,cover-thomas}, without regard to nonanticipation.

On the other hand, filtering theory is developed by imposing real-time realizability on the estimators with respect to measurement data.  Although, both reliable communication and filtering (state estimation for control) are concerned with reconstruction of processes, the main underlying assumptions characterizing them are different.
\par In this paper, the intersection of rate distortion function (RDF) and real-time realizable filtering theory is established by invoking a nonanticipative constraint on the reconstruction conditional distribution to be realizable via real-time operations, while the optimal nonanticipative reconstruction conditional distribution is derived. Consequently, the connection between nonanticipative RDF, its characterization via the optimal reconstruction conditional distribution, and real-time realizable filtering theory is established under very general conditions on the source (including Markov sources).
\par The fundamental advantage of the new filtering approach based on nonanticipative RDF, is the ability to ensure average or probabilistic estimation error constraints, which is  non-trivial task if Bayesian filtering techniques are employed to formulate such constraints. The motivations includes nonanticipative data compression over noisy channels, such as control over networks, where the controlled system and controller may be connected via a noisy channel \cite{yuksel2011,gupta-dana-hespanha-murray-hassibi2009,freudenberg-middleton2008,borkar-mitter-tatikonda2001,tatikonda-mitter2004,yuksel-meyn2012}. In such applications, filtering via nonanticipative RDF approximates sensor measurements by the reconstruction process taking values in a set of smaller cardinality, while the approximation is quantified by the distortion function. Given the recent interest in developing controller and estimator architectures processing quantized information and, in general, communication schemes for control applications, nonanticipative RDF is necessary for developing zero-delay or limited delay quantization schemes. Moreover, nonanticipative RDF is necessary for the realization of the compression channel by communication systems processing information causally.
\par The first relation between information theory and filtering via distortion rate function is discussed by R. S. Bucy in \cite{bucy}, by carrying out the computation of a realizable (nonanticipative) distortion rate function with square criteria for two samples of the Ornstein-Uhlenbeck Gaussian process. Related work on nonanticipative rate distortion theory is pursued by A.~K.~Gorbunov and M.~S.~Pinsker in \cite{gorbunov-pinsker,gorbunov91}. Specifically, \cite{gorbunov-pinsker} discussed nonanticipative RDF for general stationary processes and establishes existence of the infinite horizon limit, while \cite{gorbunov91} computes a closed form expression for nonanticipative RDF (called $\epsilon$-entropy) for stationary Gaussian processes using power spectral methods. Further elaborations on the similarities and differences between \cite{bucy,gorbunov-pinsker,gorbunov91} and this paper will be discussed in subsequent parts of the paper. Moreover, over the years several papers appeared in the literature in which controller or estimator are designed based on information theoretic measures \cite{galdos-gustafson1977,feng-loparo-fang1997,guo-yin-wang-chai2009}. An earlier work designing filters via information theoretic measures is \cite{teneketzis-thesis}, while \cite{caines1988} analyzes mutual information for Gaussian processes.

%Our interest in nonanticipative RDF is motivated by the recent interest in communication for filtering and control applications, in which zero delay or limited delay compression or quantization is highly desirable. With respect to this requirement nonanticipation of the optimal compression distribution is a necessary condition for developing end-to-end communication schemes which are not subject to decoding delays. 
\par In this paper, the connection between nonanticipative rate distortion theory and filtering theory is further examined, under a nonanticipative condition defined by the family of conditional distributions (reconstructions), for general distortion functions and random processes on abstract Polish spaces. The connection is established via optimization on the space of conditional distributions with average distortion constraint and almost sure ($a.s.$) constraints to account for the nonanticipative condition on the reconstruction conditional distribution. The main results are the following.
\begin{description}
\item[(1)] Existence of the nonanticipative RDF using the topology of weak$^*$-convergence;
\item[(2)] Closed form expression for reconstruction conditional distribution minimizing the nonanticipative RDF for stationary processes;
\item[(3)] Realization procedure of the filter based on the nonanticipative RDF;
\item[(4)] Example to demonstrate the realization of the filter.
\end{description}
It is important to point out that items (1)-(4) above are not addressed in the related papers \cite{bucy,gorbunov-pinsker,gorbunov91}. Moreover, (2) together with (3) are important in reliable communication for filtering and control applications, because one can develop communication architectures which operate with zero-delay or limited delay, as opposed to the classical RDF which is anticipative.  

Next,we give a high level discussion on Bayesian filtering theory and nonanticipative RDF, and we present some aspects of the problem pursued in this paper.
Consider a discrete-time process $X^n\triangleq\{X_0,X_1,\ldots,X_n\}\in{\cal X}_{0,n} \triangleq \times_{i=0}^n{\cal X}_i$, and its reconstruction $Y^n\triangleq\{Y_0,Y_1,\ldots,Y_n\}\in{\cal Y}_{0,n} \triangleq \times_{i=0}^n{\cal Y}_i$, where ${\cal X}_i$ and ${\cal Y}_i$ are Polish spaces (complete separable metric spaces). The objective is to reconstruct $X^n$ by $Y^n$ via nonanticipative operations subject to a distortion or fidelity criterion. That is, for each $i=0,1,\ldots$, the reconstruction $Y_i$ of $X_i$ should depend on past and present information $\{X_0,Y_0,X_1,Y_1,\ldots,X_{i-1},Y_{i-1},X_i\}$. Once this mapping is found a procedure is introduced to realize the filter of $Y_i$ from auxiliary\footnote{This point is explained in Subsection~\ref{1.2}.} measurements.

\subsection{Bayesian Estimation Theory}
In classical filtering \cite{elliott-aggoun-moore1995}, one is given a mathematical model that generates the process $X^n$, via its conditional distribution $\{P_{X_i|X^{i-1}}$ $(dx_i|x^{i-1}):i=0,1,\ldots,n\}$ or via discrete-time recursive dynamics, a mathematical model that generates observed data obtained from sensors, say, $Z^n$, $\{P_{Z_i|Z^{i-1},X^i}$ $(dz_i|z^{i-1},x^i):i=0,1,\ldots,n\}$, while $Y^n$ are the causal estimates of some function of the process $X^n$ based on the observed data $Z^n$. Note that for a memoryless channel that generates the observation sequence $\{Z_i:~i=0,1,\ldots,n\}$ then $P_{Z_i|Z^{i-1},X^i}(dz_i|z^{i-1},x^i)=P_{Z_i|X_{i}}(dz_i|x_i)-a.s.,~i=0,1,\ldots,n$.\\
In Bayesian estimation one is interested in causal estimators of some function $\Phi:{\cal X}_n\longmapsto\mathbb{R}$, $Y_n\triangleq\Phi(X_n)$ based on the observed data $Z^{n-1}\triangleq\{Z_0,Z_1,\ldots,Z_{n-1}\}$. With respect to minimizing the least-squares error pay-off, the best estimate of ${\Phi}(X_i)$ given $Z^{i-1}$, denoted by $\widehat{\Phi}(X_i)$, is given by the conditional mean
\begin{eqnarray*}
\widehat{\Phi}(X_i)\triangleq\mathbb{E}\Big\{\Phi(X_i)|Z^{i-1}\Big\}=\int_{{\cal X}_i}\Phi(x)P_{X_i|Z^{i-1}}(dx|z^{i-1}),~i=0,1,\ldots,n.
\end{eqnarray*}
For non-linear problems, Bayesian filtering is often addressed via the conditional distribution $\{P_{X_i|Z^{i-1}}(dx_i|z^{i-1}):i=0,1,\ldots,n\}$ or its unnormalized versions which satisfy discrete-recursions \cite{elliott-aggoun-moore1995}, and forms a sufficient statistic for the filtering problem.\\
Consider the simplified example of the multi-dimensional Gaussian-Markov processes modeled by
\begin{eqnarray}
\left\{ \begin{array}{ll} X_{k+1}=A_kX_k+B_kW_k,~X_0{\sim}N(0;~\Sigma_{x_0}),~k=0,1,\ldots,n-1\\
Z_k=C_kX_k+D_kV_k,~k=0,1,\ldots,n \end{array} \right.\label{equation50}
\end{eqnarray}
where $\{A_k, B_k, C_k, D_k\}$ are time-varying matrices having appropriate dimensions, $W_k{\sim}N(0;\Sigma_{W_k})$ (Gaussian with mean zero and covariance $\Sigma_{W_k}$), $V_k{\sim}N(0;\Sigma_{V_k})$, $k=0,1,\ldots,n$, while the processes $\{W_k:~k=0,1,\ldots,n-1\}, \{V_k:~k=0,1,\ldots,n\}$ are mutually independent, and independent of $X_0$. The classical Kalman Filter \cite{elliott-aggoun-moore1995} is a well-known example for which the optimal reconstruction $\widehat{X_i} =\mathbb{E}[X_i | Z^{i-1}],~i=0,1,\ldots,n$, is the conditional mean which minimizes the average  least-squares estimation error. Thus, in classical filtering theory both models which generate the unobserved and observed processes, $X^n$ and $Z^n$, respectively, are given \'a priori, and the estimator $\widehat{X}_i$ is a nonanticipative function of the past information $Z^{i-1},~i=0,1,\ldots,n$. Fig.~\ref{filtering} illustrates the cascade block diagram of the Bayesian filtering problem.	
\begin{figure}[ht]
\centering
\includegraphics[scale=0.70]{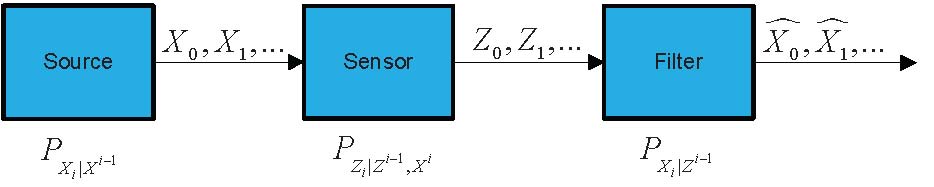}
\caption{Block Diagram of Bayesian Filtering Problem}
\label{filtering}
\end{figure}

\subsection{Nonanticipative Rate Distortion Theory and Estimation}\label{1.2}
\par In nonanticipative rate distortion theory one is given the process $X^n$, which induces  the conditional distributions $\{P_{X_i|X^{i-1}}(dx_i|x^{i-1}):~i=0,1,\ldots,n\}$ and determines the nonanticipative reconstruction conditional distribution $\{P_{Y_i|Y^{i-1},X^i}(dy_i|y^{i-1},x^i):~i=0,1,\ldots,n\}$ which minimizes the mutual information between $X^n$ and $Y^n$ subject to a distortion or fidelity constraint, via a nonanticipative or realizability constraint. The filter or estimate $\{Y_i:~i=0,1,\ldots,n\}$ of $\{X_i:~i=0,1,\ldots,n\}$ is found by realizing the reconstruction distribution $\{P_{Y_i|Y^{i-1},X^i}(dy_i|y^{i-1}$, $x^i):~i=0,1,\ldots,n\}$ via a cascade of sub-systems as shown in Fig.~\ref{filtering_and_causal}.
\begin{figure}[ht]
\centering
\includegraphics[scale=0.70]{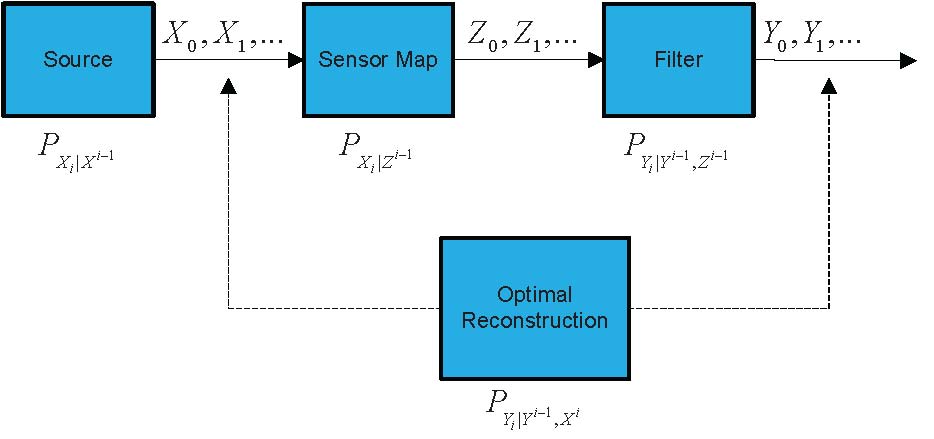}
\caption{Block Diagram of Filtering via Nonanticipative Rate Distortion Function}
\label{filtering_and_causal}
\end{figure}
The point to be made here is that the auxiliary random sequence $\{Z_0,Z_1,\ldots\}$ which is the analogue of sensor measurements (in the above discussion of Bayesian estimation) is identified during the realization of the optimal reconstruction distribution $\{P_{Y_i|Y^{i-1},X^i}(dy_i|y^{i-1},x^i):~i=0,1,\ldots,n\}$. Thus, in Bayesian estimation, the sensor map is given \'a priori, while in nonanticipative rate distortion theory, this map is identified during the realization of the optimal reconstruction conditional distribution $\{P_{Y_i|Y^{i-1},X^i}(dy_i|y^{i-1},x^i):~i=0,1,\ldots,n\}$, so that the end-to-end nonanticipative RDF from $X^n$ to $Y^n$ is achieved.

The precise problem formulation of nonanticipative RDF is defined by first introducing the distortion or fidelity constraint and mutual information.
The distortion function \cite{berger} or fidelity constraint between  $x^n$ and its reconstruction $y^n$, is a measurable function defined by
\begin{eqnarray}
d_{0,n} : {\cal X}_{0,n} \times {\cal Y}_{0,n} \mapsto [0, \infty], \: \: d_{0,n}(x^n,y^n)\triangleq\sum^n_{i=0}\rho_{0,i}(x^i,y^i).\nonumber
\end{eqnarray}
For single letter distortion $d_{0,n}(x^n,y^n)\equiv\sum_{i=0}^n\rho(x_i,y_i)$, and for single letter square error distortion $d_{0,n}(x^n,y^n)\equiv\sum_{i=0}^n{||x_i-y_i||}^2$ \cite{berger}. Moreover, for finite alphabet spaces ${\cal X}_i$ and ${\cal Y}_i$, the distortion function can be defined in terms of the Hamming distance \cite{cover-thomas}.\\
The mutual information between $X^n$ and $Y^n$, for a given distribution ${P}_{X^n}(dx^n)$, and  conditional distribution $P_{Y^n|X^n}(dy^n|x^n)$, is defined by\footnote{The precise definition of a convolution of measures denoted by $\otimes$ is given in Section~\ref{problem_formulation}.} \cite{berger}
\begin{equation}
{I}(X^n;Y^n)\triangleq \int_{{\cal X}_{0,n}\times{\cal Y}_{0,n}}\log\Big(\frac{P_{Y^n|X^n}(dy^n|x^n)}{{P}_{Y^n}(dy^n)}\Big)
P_{Y^n|X^n}(dy^n|x^n)\otimes{P}_{X^n}(dx^n). \label{1}
\end{equation}
Next, introduce the nonanticipative constraint on the reconstruction distribution. To this end, define the $(n+1)-$fold nonanticipative convolution measure
\begin{equation}
{\overrightarrow P}_{Y^n|X^n}(dy^n|x^n) \triangleq \otimes^n_{i=0}P_{Y_i|Y^{i-1},X^i}(dy_i|y^{i-1},x^i)-a.s. \label{9}
\end{equation}
The set of nonanticipative reconstruction distributions is defined by
\begin{eqnarray}
{\overrightarrow Q}_{ad} \triangleq&\Big\{ P_{Y^n|X^n}(dy^n|x^n) :~P_{Y^n|X^n}(dy^n|x^n) ={\overrightarrow P}_{Y^n|X^n}(dy^n|x^n)-a.s. \Big\}.  \label{eq18}
\end{eqnarray}
Note that without the nonanticipative constraint specified by ${\overrightarrow Q}_{ad}$, the connection between filtering theory and rate distortion theory cannot be established, since in general by Bayes' rule $P_{Y^n|X^n}(dy^n|x^n)=\otimes_{i=0}^n{P}_{Y_i|Y^{i-1},X^n}(dy_i|y^{i-1},x^n)-a.s.$, and hence, for each $i=0,1, \ldots,n$, the conditional distribution ${P}_{Y_i|Y^{i-1},X^n}(\cdot|\cdot,\cdot)$ of $Y_i$ will depend on future symbols $\{X_{i+1},X_{i+2},\ldots,X_n\}$, in addition to the past and present symbols $\{Y^{i-1},X^i\}$. However, by imposing the nonanticipative constraint (\ref{eq18}), then at each time instant $i=0,1,\ldots,$ the reconstruction $Y_i$ of $X_i$ will depend on the past reconstructions $\{Y_0,\ldots,Y_{i-1}\}$ and past and present symbols $\{X_0,\ldots,X_i\}$. For filtering and control applications, the nonanticipative constraint is necessary to avoid anticipative processing of information, while for quantization or compression applications it offers the possibility to realize the compression channel (optimal reconstruction distribution)  via causal operations and achieve an end-to-end compression with zero-delay.   

{\it Nonanticipative Distortion Rate Function.} The nonanticipative distortion rate function is defined by the minimization over ${ P}_{Y^n|X^n}(dy^n|x^n)$ of the average distortion function subject to a constraint on the mutual information rate $I(X^n;Y^n)\leq R$ and the nonanticipative constraint (\ref{eq18}) as follows.
\begin{eqnarray}
{D}^{na}_{0,n}(R)\triangleq \inf_{{ P}_{Y^n|X^n}(dy^n|x^n)\in \overrightarrow{Q}_{ad}: I(X^n ;{Y}^n)\leq R}\mathbb{E}\Big\{d_{0,n}(X^n,Y^n)\Big\}.\label{6}
\end{eqnarray}
The classical distortion rate function does not imposes the nonanticipative constraint $P_{Y^n|X^n}(dy^n|$\\
$x^n)={\overrightarrow P}_{Y^n|X^n}(dy^n|x^n)-a.s.$, hence the resulting optimal reconstruction distribution of symbol $y_i$ will depend on $(y^{i-1},x^i)$ and on future symbols $(x_{i+1},\ldots,x_n)$. Thus, by solving (\ref{6}) and then realizing the conditional distribution the optimal causal filter will be defined.\\
At this stage it is important to point out that the nonanticipative condition (\ref{eq18}) is different from the realizability condition  in \cite{bucy}, in which is assumed that $Y_i$ is independent of $X_{j|i}^*\triangleq X_j -\mathbb{E}\Big(X_j|X^i\Big), j=i+1, i+2, \ldots,$. Moreover, the nonanticipative condition (\ref{eq18}) is implied by the nonanticipative condition found in \cite{gorbunov-pinsker,gorbunov91}, defined by $X_{n+1}^\infty \leftrightarrow X^n \leftrightarrow Y^n$ forms a Markov chain for any $n=0,1,\ldots$ (e.g., $P_{Y^n|X^n,X_{n+1}^{\infty}}(dy^n|x^n,x_{n+1}^{\infty})=P_{Y^n|X^n}(dy^n|x^n),~n=0,1,\ldots$).  The claim here is that the nonanticipative condition (\ref{eq18}) is more natural and applies to processes which are not necessarily Gaussian with square error distortion function, and subject to slight modification to controlled sources in which the control is a function of the reconstruction process (we shall discuss this point further).

 {\it Nonanticipative Rate Distortion Function.}
An equivalent problem to (\ref{6}) is the nonanticipative RDF defined by
\begin{equation}
{R}^{na}_{0,n}(D)\triangleq \inf_{{ P}_{Y^n|X^n}(dy^n|x^n)\in  \overrightarrow{Q}_{ad}:~{E}\big\{d_{0,n}(X^n,Y^n)\leq{D}\big\}}I(X^n; Y^n).\label{7}
\end{equation}
The two problems defined by (\ref{6}) and (\ref{7}) are equivalent in the sense that the solution of (\ref{6}) gives that of (\ref{7}) and vice-versa \cite{cover-thomas}. Moreover, it can be shown that
\begin{eqnarray}
{P}_{Y^n|X^n}(dy^n|x^n)&=&{\overrightarrow P}_{Y^n|X^n}(dy^n|x^n)-a.s.
{\Longleftrightarrow} \nonumber\\
I(X^n;Y^n)&=&\int_{{\cal X}_{0,n}\times{\cal Y}_{0,n}} \log\Big(\frac{{\overrightarrow P}_{Y^n|X^n}(dy^n|x^n)}{{P}_{Y^n}(dy^n)}\Big){\overrightarrow P}_{Y^n|X^n}(dy^n|x^n)\otimes{P}_{X^n}(dx^n)\nonumber\\
&\equiv&{\mathbb I}(P_{X^n},{\overrightarrow P}_{Y^n|X^n})\label{eq6}
\end{eqnarray}
where the notation ${\mathbb I}(P_{X^n},{\overrightarrow P}_{Y^n|X^n})$ is used to point out the functional dependence of $I(X^n;{Y^n})$ on $\{P_{X^n},{\overrightarrow P}_{Y^n|X^n}\}$. The nonanticipative distortion rate function and RDF can be generalized to controlled sources.
\par The paper is organized as follows. Section~\ref{problem_formulation} discusses the problem formulation on abstract spaces. Section~\ref{existence} establishes  existence of optimal minimizing reconstruction distribution, and Section~\ref{necessary} derives the stationary solution. Section~\ref{realization} describes the real-time realization of nonanticipative RDF. Finally, Section~\ref{example} demonstrates the filter realization via an example.

%%%%%%%%%%%%%%%%%%%%%%%%%%%%%%%%%%%%%%%%%%%%%%%%%%%%%%%%%%%%%%%%%%%%%%%%%%%%%%%%%%%%%%%%%%%

\section{Formulation of Nonanticipative RDF on Abstract Spaces}\label{problem_formulation}
\noi Throughout the paper we use the notation defined on Table~\ref{notations}.
\begin{table}[htbp]\caption{Table of Notation}
%\begin{center} 
\centering
\begin{tabular}{r c p{8cm} }
\toprule
$\mathbb{N}\triangleq\{0,1,\ldots\}$ &~& Set of nonnegative integers\\
$\mathbb{N}^n\triangleq\{0,1,\ldots,n\}$ &~& Set of first $(n+1)$ nonnegative integers\\
${\cal X}_t$,~${\cal Y}_t$,~$t\in\mathbb{N}$ &~& Source and reconstruction alphabets\\
${\cal B}({\cal X}_t)$,~${\cal B}({\cal Y}_t)$ &~& $\sigma$-algebra of Borel sets generated by ${\cal X}_t$, ${\cal Y}_t$\\
${\cal X}_{0,n}\triangleq {\times}_{k=0}^{n}{\cal X}_k$,~${\cal Y}_{0,n}\triangleq {\times}_{k=0}^{n}{\cal Y}_k$ &~& Cartesian product of source and reconstruction alphabets\\
${\cal B}({\cal X}_{0,n})\triangleq\times_{k=0}^n{\cal B}({\cal X}_k)$,~${\cal B}({\cal Y}_{0,n})\triangleq\times_{k=0}^n{\cal B}({\cal Y}_k)$ &~& $\sigma$-algebra of Borel sets generated by ${\cal X}_{0,n}$, ${\cal Y}_{0,n}$\\
$x^n \triangleq \{x_0,\ldots,x_n\}$,~$y^n \triangleq \{y_0,\ldots,y_n\}$ &~& Sequence of source and reconstruction symbols\\
${\cal M}_1({\cal Z})$ &~& Set of probability measures on a measurable space $({\cal Z}, {\cal B}({\cal Z}))$\\
${\cal Q}({\cal Y};{\cal X})$ &~& Set of stochastic kernels on $({\cal Y},{\cal B}({\cal Y}))$ given $({\cal X},{\cal B}(\cal X))$\\
$X\leftrightarrow Y \leftrightarrow Z\Leftrightarrow{P}_{Z|X,Y}(dz|x,y)=P_{Z|Y}(dz|y)-a.s.$ &~& Markov Chain (MC) or conditional independence\\
$BC({\cal X})$ &~& Vector space of bounded continuous real-valued functions defined on a Polish space ${\cal X}$\\
$L_1\big{(}\mu,BC({\cal X})\big{)}$ &~& Set of all $\mu$-integrable functions defined on ${\cal X}$ with values in $BC({\cal X})$\\
$||\cdot||_{\mu}$ &~& Norm with respect to $L_1\big{(}\mu,BC({\cal X})\big{)}$\\
${\cal X}^*$ &~& Topological dual of a Banach space ${\cal X}$\\
${M}_{rba}({\cal X})$ &~& Space of finitely additive regular bounded signed measures on $({\cal X},{\cal B}({\cal X}))$\\
${\Pi}_{rba}({\cal X})\subset{M}_{rba}({\cal X})$ &~& Space of finitely additive regular bounded probability measures on $({\cal X},{\cal B}({\cal X}))$\\
\bottomrule
\end{tabular}
%\end{center}
\label{notations}
\end{table}
\noi The source and reconstruction alphabets, respectively, are sequences of Polish spaces $\{ {\cal X}_t: t\in\mathbb{N}\}$ and $\{ {\cal Y}_t: t\in\mathbb{N}\}$, associated with their corresponding measurable spaces $({\cal X}_t,{\cal B}({\cal X}_t))$ and $({\cal Y}_t, {\cal B}({\cal Y}_t))$, $t\in\mathbb{N}$. Sequences of alphabets are  identified
with the product spaces $({\cal X}_{0,n},{\cal B}({\cal X}_{0,n})) \triangleq {\times}_{k=0}^{n}({\cal X}_k,{\cal B}({\cal X}_k))$, and $({\cal Y}_{0,n},{\cal B}({\cal Y}_{0,n}))\triangleq \times_{k=0}^{n}({\cal Y}_k,{\cal B}({\cal Y}_k))$. The source and reconstruction are random processes denoted by $X^n \triangleq \{X_t: t\in\mathbb{N}^n\}$, $X:~\{t\}\times\Omega\mapsto {\cal X}_t$, and by $Y^n\triangleq \{Y_t: t\in\mathbb{N}^n\}$, $Y:~\{t\}\times\Omega\mapsto  {\cal Y}_t$, respectively.\\ %Probability measures on any measurable space $( {\cal Z}, {\cal B}({\cal Z}))$ are denoted by ${\cal M}_1({\cal Z})$.\\ 
\noi The reconstruction conditional distribution will be defined via stochastic kernels. Note that the random variable (RV) ${Z}$ is called conditional independent of RV $X$ given the RV $Y$ if and only if $X\leftrightarrow Y \leftrightarrow Z$ forms a MC in both directions, equivalently $P_{X,Z|Y}(dx,dz|y)=P_{X|Y}(dx|y)P_{Z|Y}(dz|y)-a.s.$, or $P_{Z|Y,X}(dz|y,x)=P_{Z|Y}(dz|y)-a.s.$

\begin{definition}\label{stochastic kernel}\cite{dupuis-ellis97}
Let $({\cal X}, {\cal B}({\cal X})), ({\cal Y}, {\cal B}({\cal Y}))$ be measurable spaces in which $\cal Y$ is a Polish Space.\\
A  stochastic kernel on $\cal Y$ given $\cal X$ is a mapping $q: {\cal B}({\cal Y}) \times {\cal X}  \rightarrow [0,1]$ satisfying the following two properties:
\begin{enumerate}
\item[(1)] For every $x \in {\cal X}$, the set function $q(\cdot;x)$ is a probability measure (possibly finitely additive) on ${\cal B}({\cal Y})$;
\item[(2)] For every $F \in {\cal B}({\cal Y})$, the function $q(F;\cdot)$ is ${\cal B}({\cal X})$-measurable.
\end{enumerate}
%The set of all such stochastic kernels is denoted by ${\cal Q}({\cal Y};{\cal X})$.
\end{definition}
%\par An important notion which is used in nonanticipative RDF is conditional independence.
\noi Stochastic kernels can be used to define anticipative and nonanticipative convolution of reconstruction kernels and associated classical and nonanticipative RDF.
\begin{definition}\label{comprchan}
Given measurable spaces $({\cal X}_{0,n},{\cal B}({\cal X}_{0,n}))$, $({\cal Y}_{0,n},{\cal B}({\cal Y}_{0,n}))$, and their product spaces, data compression channels are classified as follows.
\begin{enumerate}
\item[1)] {\it An Anticipative Data Compression Channel} is a  stochastic kernel $ q_{0,n} (dy^n; x^n) \in {\cal Q}({\cal Y}_{0,n} ;{\cal X}_{0,n})$. Such a kernel admits a factorization into a convolution of a sequence of anticipative stochastic kernels as follows
\begin{equation}
q_{0,n}(dy^n; x^n)=\otimes_{i=0}^n q_i(dy_i;y^{i-1},x^n)-a.s.\label{eq.1}
\end{equation}
where  $q_i(dy_i;y^{i-1},x^n) \in {\cal Q}({\cal Y}_i;{\cal Y}_{0,i-1}\times{\cal X}_{0,n}), i=0,\ldots,n,~n \in \mathbb{N}$.
\item[2)]{\it A Nonanticipative Convolution Data Compression Channel} is a convolution of a sequence of nonanticipative stochastic kernels defined by
\begin{equation}
{\overrightarrow q}_{0,n}(dy^n;x^n) \triangleq \otimes_{i=0}^n q_i(dy_i;y^{i-1},x^i)-a.s.\label{eq.2}
\end{equation}
where  $q_i(dy_i;y^{i-1},x^i) \in {\cal Q}({\cal Y}_i;{\cal Y}_{0,i-1}\times{\cal X}_{0,i}), i=0,\ldots,n,~n \in \mathbb{N}$.
\item[3)]{\it A Restricted Nonanticipative Data Compression Channel} is a stochastic kernel $q_{0,n} (dy^n; x^n)$ $\in {\cal Q}({\cal Y}_{0,n} ;{\cal X}_{0,n})$ which is a convolution of a sequence of nonanticipative stochastic kernels obtained by imposing the almost sure (a.s.) constraint defined by
\begin{equation}
q_{0,n}(dy^n; x^n)=\otimes_{i=0}^n q_i(dy_i;y^{i-1},x^i)-a.s.\label{eq.3}
\end{equation}
where $q_i \in {\cal Q}({\cal Y}_i;{\cal Y}_{0,i-1}\times{\cal X}_{0,i}), i=0,\ldots,n,~n \in \mathbb{N}$.
\end{enumerate}
\end{definition}
\noi As stated earlier, the classical RDF is concerned with optimizing (\ref{1}) with respect to anticipative stochastic kernels (\ref{eq.1}). In this paper we address problem (\ref{6}) or (\ref{7}), i.e., when the conditional distribution (stochastic kernel) is restricted nonanticipative, and we discuss  generalizations based on (\ref{eq.2}). That is, for a given distribution $P_{X^n}(dx^n)$, nonanticipative RDF imposes the a.s.-constraint (\ref{eq.3}) on the reconstruction conditional distribution, and hence on the joint distribution $P_{X^n,Y^n}(dx^n,dy^n)$ generated by them, unlike the classical RDF which does not impose such a constraint. However, when the source is independently distributed, i.e., $P_{X^n}(dx^n)=\otimes_{i=0}^n{P}_{X_i}(dx_i)-a.s.$, it is well known that the optimal reconstruction conditional distribution of the classical RDF has the property $P^*_{Y^n|X^n}(dy^n|x^n)=\otimes_{i=0}^n{P}^*_{Y^i|X^i}(dy_i|x_i)-a.s.$ It is also well known that for sources with memory (i.e., Markov sources) the optimal reconstruction conditional distribution of the classical RDF is anticipative, i.e., $P^*_{Y^n|X^n}(dy^n|x^n)=\otimes_{i=0}^n{P}^*_{Y_i|Y^{i-1},X^n}(dy_i|y^{i-1},x^n)-a.s.$ Therefore, to ensure a nonanticipative reconstruction conditional distribution one has to impose the constraint (\ref{eq.3}) to the classical RDF. On the other hand, a nonanticipative convolution data compression channel (\ref{eq.2}) does not impose any constraint on the joint distribution $P_{X^n,Y^n}(dx^n,dy^n)$. This point is further explained below by discussing generalizations of distortion rate function (\ref{6}) and RDF (\ref{7}).
\begin{remark}\label{generalizations}
The nonanticipative distortion rate function and the nonanticipative RDF can be generalized as follows. Given a sequence of conditional distributions $\{{P}_{X_i|X^{i-1},Y^{i-1}}(dx_i|x^{i-1},y^{i-1})$ $:~i=0,1,\ldots,n\}$ then (\ref{6}) and (\ref{7}) can be generalized to
\begin{eqnarray}
{\overrightarrow D}_{0,n}^{na}(R)&\triangleq \inf_{{\overrightarrow P}_{Y^n|X^n}(dy^n|x^n):I(X^n\rightarrow{Y^n})\leq R}E\Big\{d_{0,n}(X^n,Y^n)\Big\}\label{eq21}\\
{\overrightarrow R}_{0,n}^{na}(D)&\triangleq \inf_{{\overrightarrow P}_{Y^n|X^n}(dy^n|x^n):E\{d_{0,n}(X^n,Y^n)\leq{D}\}}I(X^n\rightarrow{Y^n})\label{eq22}
\end{eqnarray}
where $I(X^n\rightarrow{Y^n})$ is the directed information measure from $X^n$ to $Y^n$ defined by
\begin{eqnarray}
{I}(X^n\rightarrow{Y}^n)&\triangleq& \int_{{\cal X}_{0,n}\times{\cal Y}_{0,n}}\log\Big(\frac{\overrightarrow{P}_{Y^n|X^n}(dy^n|x^n)}{{P}_{Y^n}(dy^n)}\Big)
\overrightarrow{P}_{Y^n|X^n}(dy^n|x^n)\otimes\overleftarrow{P}_{X^n|Y^{n-1}}(dx^n|y^{n-1})\nonumber\\
&\equiv&\mathbb{I}_{X^n\rightarrow{Y^n}}(\overleftarrow{P}_{X^n|Y^{n-1}},\overrightarrow{P}_{Y^n|X^n})\nonumber
\end{eqnarray}
where $\overleftarrow{P}_{X^n|Y^{n-1}}$ is defined by
\begin{eqnarray}
\overleftarrow{P}_{X^n|Y^{n-1}}(dx^n|y^{n-1})\triangleq\otimes_{i=0}^n{P}_{X_i|X^{i-1},Y^{i-1}}(dx_i|x^{i-1},y^{i-1})-a.s.\nonumber
\end{eqnarray}
Clearly, (\ref{eq21}) and (\ref{eq22}) do not assume $P_{X_i|X^{i-1},Y^{i-1}}(dx_i|x^{i-1},y^{i-1})=P_{X_i|X^{i-1}}(dx_i|x^{i-1})-a.s.$, and hence the process $X^n$ can be affected by $Y^n$ causally. It is easy to show that if (\ref{eq.3}) holds then $P_{X_i|X^{i-1},Y^{i-1}}(dx_i|x^{i-1},y^{i-1})=P_{X_i|X^{i-1}}(dx_i|x^{i-1})-a.s.,~i=0,1,\ldots,n$, also holds, and hence (\ref{eq21}), (\ref{eq22}) reduce to (\ref{6}), (\ref{7}). The generalizations (\ref{eq21}), (\ref{eq22}), covers conditionally Gaussian sources as a special case \cite{liptser-shiryaev1978}. It also covers the case when the source is a controlled process, controlled over a finite rate channel based on the quantized or reconstruction process. These generalizations will be investigated elsewhere, since they will require new topological spaces on which existence of optimal solution to (\ref{eq21}) and (\ref{eq22}) can be shown.
\end{remark}
\subsection{Nonanticipative RDF}
In this subsection the nonanticipative RDF is rigorously defined on abstract spaces. Given a source probability measure ${\cal \mu}_{0,n} \in {\cal M}_1({\cal X}_{0, n})$ (possibly finitely additive) and a reconstruction kernel $q_{0,n} \in {\cal Q}({\cal Y}_{0, n};{\cal X}_{0, n})$, one can define three probability measures as follows.
\par (P1): The joint measure $P_{0,n} \in {\cal M}_1({\cal Y}_{0,n}\times {\cal X}_{0, n})$:
\begin{eqnarray}
P_{0,n}(G_{0,n})&\triangleq&(\mu_{0,n} \otimes q_{0,n})(G_{0,n}),\:G_{0,n} \in {\cal B}({\cal X}_{0,n})\times{\cal B}({\cal Y}_{0,n})\nonumber\\
&=&\int_{{\cal X}_{0,n}} q_{0,n}(G_{0,n,x^n};x^n) \mu_{0,n}(d{x^n})\nonumber
\end{eqnarray}
where $G_{0,n,x^n}$ is the $x^n-$section of $G_{0,n}$ at point ${x^n}$ defined by $G_{0,n,x^n}\triangleq \{y^n \in {\cal Y}_{0,n}: (x^n, y^n) \in G_{0,n}\}$ and $\otimes$ denotes the convolution.
\par (P2): The marginal measure $\nu_{0,n} \in {\cal M}_1({\cal Y}_{0,n})$:% corresponding to the kernel $q_{0,n} \in {\cal Q}({\cal Y}_{0, n};{\cal X}_{0, n})$:
\begin{eqnarray}
\nu_{0,n}(F_{0,n})&\triangleq& P_{0,n}({\cal X}_{0, n} \times F_{0,n}),~F_{0,n} \in {\cal B}({\cal Y}_{0,n})\nonumber\\
&=&\int_{{\cal X}_{0, n}} q_{0,n}(({\cal X}_{0, n}\times F_{0,n})_{{x}^{n}};{x}^{n}) \mu_{0,n}(d{x^n})=\int_{{\cal X}_{0, n}} q_{0,n}(F_{0,n};x^n) \mu_{0,n}(dx^n).\nonumber
\end{eqnarray}
\par(P3): The product measure  $\pi_{0,n}:{\cal B}({\cal X}_{0,n}) \times
{\cal B}({\cal Y}_{0,n}) \mapsto [0,1] $ of $\mu_{0,n}\in{\cal M}_1({\cal X}_{0, n})$ and $\nu_{0,n}\in{\cal M}_1({\cal Y}_{0, n})$ for $G_{0,n} \in {\cal B}({\cal X}_{0,n}) \times {\cal B}({\cal Y}_{0,n})$:
\begin{eqnarray}
\pi_{0,n}(G_{0,n})\triangleq(\mu_{0,n} \times \nu_{0,n})(G_{0,n})=\int_{{\cal X}_{0, n}} \nu_{0,n}(G_{0,n,x^n}) \mu_{0,n}(dx^n).\nonumber
\end{eqnarray}
\noi The precise definition of mutual information between two sequences of Random Variables $X^n$ and $Y^n$, denoted $I(X^n; Y^n)$ is defined via the Kullback-Leibler distance (or relative entropy) between the joint probability distribution of $(X^n, Y^n)$ and the product of its marginal probability distributions of $X^n$ and $Y^n$, using the Radon-Nikodym derivative as follows.
\begin{definition}\label{relative_entropy}\cite{ihara1993}
Given a measurable space $({\cal X},  {\cal B}({\cal X}))$, the relative entropy between two probability measures $P, Q\in {\cal M}_1({\cal X})$ is defined by
\begin{eqnarray}
\mathbb{D}(P || Q) \triangleq \left\{ \begin{array}{ll} \int_{{\cal X}} \log   \Big(  \frac{dP}{dQ}\Big) dP = \int_{{\cal X}} \log \Big(\frac{dP}{dQ}\Big)  \frac{dP}{dQ} dQ & \mbox{if} \quad P << Q \\
+ \infty &  \mbox{otherwise} \end{array} \right.
\nonumber
\end{eqnarray}
where $ \frac{dP}{dQ}$ denotes the Radon-Nikodym derivative (density) of  $P$ with respect to  $Q$, and $P << Q$ denotes absolute continuity of  $Q$ with respect to $P$.
\end{definition}
Hence, by the construction of probability measures (P1)-(P3), and the chain rule of relative entropy \cite{dupuis-ellis97}, the following equivalent definitions of mutual information are obtained.
\begin{eqnarray}
I(X^n;Y^n) &\triangleq& \mathbb{D}(P_{0,n}|| \pi_{0,n})\label{re4}\\
&=&\int_{{\cal X}_{0,n} \times {\cal Y}_{0,n}}\log \Big( \frac{d  (\mu_{0,n} \otimes q_{0,n}) }{d ( \mu_{0,n} \times \nu_{0,n} ) }\Big) d(\mu_{0,n} \otimes q_{0,n}) \label{eq.4}\\
& =& \int_{{\cal X}_{0,n} \times {\cal Y}_{0,n}} \log \Big( \frac{q_{0,n}(d y^n; x^n)}{  \nu_{0,n} (dy^n)   } \Big)q_{0,n}(dy^n;dx^n)\otimes\mu_{0,n}(dx^n) \label{eq.5}\\
&=&\int_{{\cal X}_{0,n}} \mathbb{D}(q_{0,n}(\cdot;x^n)|| \nu_{0,n}(\cdot)) \mu_{0,n}(dx^n)\nonumber\\
&\equiv& \mathbb{I}(\mu_{0,n}, q_{0,n}).  \label{re3}
\end{eqnarray}
Note that (\ref{re3}) states that mutual information is expressed as a functional of $\{\mu_{0,n}, q_{0,n}\}$ and it is denoted by $\mathbb{I}(\mu_{0,n},q_{0,n})$. Note also that $\mu_{0,n}\otimes{q}_{0,n}\ll{\mu}_{0,n}\times\nu_{0,n}$ if and only if $q(\cdot;x^n)\ll{\nu}_{0,n}(\cdot)$, $\mu_{0,n}-$a.s., which is used to established that (\ref{eq.4}) is equivalent to (\ref{eq.5}). Necessary and sufficient conditions for existence of  a Radon-Nikodym derivative  for finitely additive measures can be found in \cite{maynard79}.\\
Next, the classical RDF \cite{berger} is introduced, since the definition of nonanticipative RDF will be based on the classical definition by imposing the nonanticipative constraint (\ref{eq18}) (or Definition~\ref{comprchan}-3).
\begin{definition}\label{classical_rdf}
\cite{berger}$(${\bf Classical Rate Distortion Function}$)$ Let $d_{0,n}: {\cal X}_{0,n}  \times {\cal Y}_{0,n} \rightarrow [0, \infty]$, be an ${\cal B}({\cal X}_{0,n}) \times {\cal B }( {\cal Y}_{0,n})$-measurable distortion function, and let $Q_{0,n}(D) \subset {\cal Q}({\cal Y}_{0,n}; {\cal X}_{0,n})$ (assuming is nonempty) denotes the average distortion or fidelity constraint defined by
\begin{align}
Q_{0,n}(D)\tri\Big\{ q_{0,n} \in {\cal Q}({\cal Y}_{0,n}; {\cal X}_{0,n}):\int_{{\cal X}_{0,n}\times{\cal Y}_{0,n}} d_{0,n}(x^n,y^n) q_{0,n}(dy^n;x^n)\otimes\mu_{0,n}(dx^n) \leq D \Big\}\label{dc1}
\end{align}
for $D\geq 0$. The classical RDF associated with the anticipative kernel $q_{0,n} \in {\cal Q}({\cal Y}_{0,n}; {\cal X}_{0,n})$ is defined by
\begin{equation}
R_{0,n}(D) \triangleq \inf_{q_{0,n} \in Q_{0,n}(D)}\mathbb{I}(\mu_{0,n},q_{0,n}).  \label{f3s}
\end{equation}
\end{definition}
Existence in (\ref{f3s}) is shown by assuming $d_{0,n}(x^n,\cdot)$ is bounded continuous on ${\cal Y}_{0,n}$ while ${\cal Y}_{0,n}$ is compact, using weak-convergence of probability measures in \cite{csiszar74}, and for more general conditions $d_{0,n}(x^n,\cdot)$ which is only continuous on ${\cal Y}_{0,n}$ using weak$^*$-convergence of measures on Polish spaces \cite{farzad06}. \\
 Unfortunately, for general sources and distortion function $d_{0,n}(x^n,y^n)$, the optimal reconstruction $q^*_{0,n}(dy^n;x^n)=\otimes^n_{i=0}q^*_i(dy_i;y^{i-1},x^n)$ is anticipative, and hence the link to filtering theory cannot be established due to dependence of $y_i$ on $(y^{i-1},x^i)$ and on future symbols $(x_{i+1},\ldots,x_n)$. This raises the question whether the classical RDF can be reformulated so that the optimal reconstruction kernel is nonanticipative. Before the definition of nonanticipative RDF we introduced a Lemma which gives insight into how classical and nonanticipative RDF are related.
\par The next lemma relates nonanticipative convolution reconstruction kernels and conditional independence.
\begin{lemma} \label{lem1}
The following are equivalent for each $n\in\mathbb{N}$.
\begin{enumerate}
\item[1)] $q_{0,n} (dy^n; x^n)={\overrightarrow q}_{0,n}(dy^n;x^n)$-a.s., (see Definition \ref{comprchan}-3).

\item[2)] For each $i=0,1,\ldots, n-1$,  $Y_i \leftrightarrow (X^i, Y^{i-1}) \leftrightarrow (X_{i+1}, X_{i+2}, \ldots, X_n)$, forms a MC.

\item[3)] For each  $i=0,1,\ldots, n-1$, $Y^i \leftrightarrow X^i \leftrightarrow X_{i+1}$ forms a MC.
\end{enumerate}
Moreover, $X_{i+1}^n\leftrightarrow{X^i}\leftrightarrow{Y^i}$, forms a MC for each $i=0,1,\ldots,n-1$, implies any of the statements 1), 2), 3).
\end{lemma}
\begin{proof}
This is straight forward hence the derivation is omitted.
\end{proof}
\vspace{0.2cm}
\noi According to Lemma~\ref{lem1}$-1)$, for a restricted nonanticipative stochastic kernel the mutual information becomes
\begin{eqnarray}
I(X^n;Y^n)&=&\int_{{\cal X}_{0,n} \times {\cal Y}_{0,n}} \log \Big( \frac{ \overrightarrow{q}_{0,n}(d y^n; x^n)}{\nu_{0,n}(dy^n)} \Big){\overrightarrow q}_{0,n}(dy^n;dx^n)\otimes\mu_{0,n}(dx^n) \nonumber  \\
&\equiv& {\mathbb I}(\mu_{0,n},\overrightarrow{q}_{0,n} )  \label{ex11}
\end{eqnarray}
where (\ref{ex11}) states that $I(X^n;Y^n)$ is a functional of $\{\mu_{0,n},{\overrightarrow q}_{0,n}\}$. Hence, nonanticipative RDF is defined by optimizing ${\mathbb I}(\mu_{0,n},{q}_{0,n})$ over ${q}_{0,n}{\in}Q_{0,n}(D)$ subject to the realizability constraint $q_{0,n}(dy^n;x^n)={\overrightarrow q}_{0,n}(dy^n;x^n)-a.s.,$ which satisfies a distortion constraint.
\begin{definition}\label{def1}
$(${\bf Nonanticipative Rate Distortion Function}$)$
Suppose $d_{0,n}(x^n,y^n)\triangleq\sum^n_{i=0}\rho_{0,i}$ $(x^i,y^i)$, where $\rho_{0,i}: {\cal X}_{0,i}  \times {\cal Y}_{0,i}\rightarrow [0, \infty]$, is a sequence of ${\cal B}({\cal X}_{0,i}) \times {\cal B }( {\cal Y}_{0,i})$-measurable distortion functions, for $i=0,1,\ldots,n$, and let $\overrightarrow{Q}_{0,n}(D)$ (assuming is nonempty) denotes the average distortion or fidelity constraint defined by
\begin{eqnarray}
\overrightarrow{Q}_{0,n}(D)\tri Q_{0,n}(D)\bigcap\Big\{q_{0,n}{\in}{\cal Q}({\cal Y}_{0,n};{\cal X}_{0,n}):&q_{0,n}(dy^n;x^n)={\overrightarrow q}_{0,n}(dy^n;x^n)-a.s.\Big\}\label{eq2}
\end{eqnarray}
The nonanticipative RDF associated with the restricted nonanticipative stochastic kernel is defined by
\begin{eqnarray}
{R}^{na}_{0,n}(D) \tri \inf_{{{q}_{0,n}\in \overrightarrow{Q}_{0,n}(D)}}{\mathbb I}(\mu_{0,n},{q}_{0,n}).\label{ex12}
\end{eqnarray}
\end{definition}
Thus, ${R}^{na}_{0,n}(D)$ is characterized by minimizing mutual information or equivalently $\mathbb{I}(\mu_{0,n},{q}_{0,n})$ over the ${Q}_{0,n}(D)$ and the nonanticipative constraint (\ref{eq18}). In the work of \cite{gorbunov-pinsker}, nonanticipative RDF is called $\epsilon$-entropy and nonanticipation is defined via $X_{i+1}^n\leftrightarrow{X^i}\leftrightarrow{Y^i}$, forms a MC for each $i=0,1,\ldots,n-1$, which implies $q_{0,n} (dy^n; x^n)={\overrightarrow q}_{0,n}(dy^n;x^n)$. Clearly, Gorbunov and Pinsker \cite{gorbunov-pinsker} nonanticipative RDF which imposes the constraint $X_{i+1}^n\leftrightarrow{X}^i\leftrightarrow{Y^i}$ forms a MC for each $i=0,1,\ldots,n-1$, implies $P_{X_{i+1}|X^{i},Y^i}(dx_{i+1}|x^i,y^i)=P_{X_{i+1}|X^i}(dx_{i+1}|x^i)-a.s.,~i=0,1,\ldots,n-1$, and hence, it does not allow the generalizations discussed in Remark~\ref{generalizations}.

%%%%%%%%%%%%%%%%%%%%%%%%%%%%%%%%%%%%%%%%%%%%%%%%%%%%%%%%%%%%%%%%%%%%%%%%%%%%%%%%%%%%%%%%

\section{Existence of Optimal Reconstruction Kernel}\label{existence}

\par In this section, appropriate topologies and function spaces are introduced and existence of the minimizing nonanticipative product kernel in (\ref{ex12}) is proved. The construction of spaces is based on \cite{farzad06}.
\subsection{Abstract Spaces}
Let $BC({\cal Y}_{0,n})$ denote the vector space of bounded continuous real valued functions defined
on the Polish space ${\cal Y}_{0,n}$. Furnished with the sup norm
topology, this is a Banach space. Denote by $L_1(\mu_{0,n}, BC({\cal Y}_{0,n}))$ the space of all $\mu_{0,n}$-integrable functions defined on  ${\cal X}_{0,n}$ with values in $BC({\cal Y}_{0,n}),$ so that for each $\phi \in L_1(\mu_{0,n}, BC({\cal Y}_{0,n}))$ its norm is defined by
\begin{eqnarray}
\parallel \phi \parallel_{\mu_{0,n}} \triangleq \int_{{\cal X}_{0,n}} ||\phi(x^n,\cdot)||_{BC({\cal Y}_{0,n})} \mu_{0,n}(dx^n) <\infty\nonumber.
\end{eqnarray}
The norm topology $\parallel{\phi}\parallel_{\mu_{0,n}}$, makes $L_1(\mu_{0,n}, BC({\cal Y}_{0,n}))$ a Banach space. The topological dual of $BC({\cal Y}_{0,n})$ denoted by $ \Big( BC({\cal Y}_{0,n})\Big)^*$ is isometrically isomorphic to the Banach space of finitely additive regular bounded signed measures on ${\cal Y}_{0,n}$ \cite{dunford1988}, denoted by $M_{rba}({\cal Y}_{0,n})$. Let $\Pi_{rba}({\cal Y}_{0,n})\subset M_{rba}({\cal Y}_{0,n})$ denote the set of regular bounded finitely additive probability measures on ${\cal Y}_{0,n}$.  Clearly if ${\cal Y}_{0,n}$ is compact, then $\Big(BC({\cal Y}_{0,n})\Big)^*$ will be isometrically isomorphic to the space of countably additive signed measures, as
in \cite{csiszar74}. It follows from the theory of ``lifting" \cite{tulcea1969} that the dual
of the space $L_1(\mu_{0,n},BC({\cal Y}_{0,n}))$ is $L_{\infty}^w(\mu_{0,n}, M_{rba}({\cal Y}_{0,n}))$, denoting the space of all $M_{rba}({\cal Y}_{0,n})$ valued functions $\{q\}$ which are weak$^*$-measurable in the sense that for each $\phi \in
BC({\cal Y}_{0,n}),$  $x^n \rightarrow q_{x^n}(\phi) \triangleq \int_{{\cal Y}_{0,n}}\phi(y^n) q(dy^n;x^n)$ is $\mu_{0,n}$-measurable and $\mu_{0,n}$-essentially
bounded.
\subsection{Weak$^*$-Compactness and Existence}
Next, we prepare to prove existence of solution to $R_{0,n}^{na}(D)$. Define an admissible set of stochastic kernels associated with classical rate distortion function by
\begin{eqnarray*}
Q_{ad}\triangleq L_{\infty}^w(\mu_{0,n}, \Pi_{rba}({\cal Y}_{0,n})) \subset L_{\infty}^w(\mu_{0,n}, M_{rba}({\cal Y}_{0,n})).
\end{eqnarray*}
Clearly, $Q_{ad}$ is a unit sphere in $L_{\infty}^w(\mu_{0,n}, M_{rba}({\cal Y}_{0,n}))$. For each $\phi{\in}L_1(\mu_{0,n}, BC({\cal Y}_{0,n}))$ we can define a linear functional on $L_{\infty}^w(\mu_{0,n}, M_{rba}({\cal Y}_{0,n}))$ by
\begin{eqnarray*}
\ell_{\phi}(q_{0,n})\triangleq\int_{{\cal X}_{0,n}}\Big( \int_{{\cal Y}_{0,n}} \phi(x^n,y^n)q_{0,n}(dy^n;x^n) \Big)\mu_{0,n}(dx^n).
\end{eqnarray*}
This is a bounded, linear and weak$^*$-continuous functional on $L_{\infty}^w(\mu_{0,n}, M_{rba}({\cal Y}_{0,n}))$ as it is shown below.
\begin{eqnarray*}
|\ell_{\phi}(q_{0,n})|&=&\bigg{|}\int_{{\cal X}_{0,n}}\Big( \int_{{\cal Y}_{0,n}} \phi(x^n,y^n)q_{0,n}(dy^n;x^n) \Big)\mu_{0,n}(dx^n)\bigg{|}\nonumber\\
&\leq&\int_{{\cal X}_{0,n}}\bigg{|}\Big( \int_{{\cal Y}_{0,n}} \phi(x^n,y^n)q_{0,n}(dy^n;x^n) \Big)\bigg{|}\mu_{0,n}(dx^n)\nonumber\\
&\leq&\int_{{\cal X}_{0,n}}||\phi(x^n,\cdot)||_{BC({\cal Y}_{0,n})}||q_{0,n}(\cdot;x^n)||_{TV}\mu_{0,n}(dx^n)\nonumber\\
&\leq&\int_{{\cal X}_{0,n}}||\phi(x^n,\cdot)||_{BC({\cal Y}_{0,n})}\mu_{0,n}(dx^n)\nonumber\\
&=&||\phi||_{L_1(\mu_{0,n}, BC({\cal Y}_{0,n}))}<\infty.
\end{eqnarray*}
So given $\phi\in L_1(\mu_{0,n},BC({\cal Y}_{0,n}))$, there exists a $c_{\phi}<\infty$ such that $||\ell_{\phi}||<c_{\phi}$. Therefore, $\ell_{\phi}$ is a bounded, linear functional on $L_{\infty}^w(\mu_{0,n}, \Pi_{rba}({\cal Y}_{0,n}))$ and hence on $L_{\infty}^w(\mu_{0,n}, M_{rba}({\cal Y}_{0,n}))$. Thus, it is continuous in the weak$^*$-sense.
\par For $d_{0,n}: {\cal X}_{0,n}  \times {\cal Y}_{0,n}\rightarrow [0, \infty)$ measurable and $d_{0,n}{\in}L_1(\mu_{0,n},BC({\cal Y}_{0,n}))$ the distortion constraint set of the classical RDF is given by
\begin{eqnarray*}
Q_{0,n}(D)\triangleq\{q{\in}Q_{ad}:\ell_{d_{0,n}}(q_{0,n}){\leq}D\}.
\end{eqnarray*}
The next result is shown in \cite{farzad06}; it utilizes the Alaoglu's theorem \cite{dunford1988}, which states that a closed and bounded subset of a weak$^*$-compact set is weak$^*$-compact. These will be used to establish existence of minimizer in $\overrightarrow{Q}_{ad}$ for the nonanticipative RDF $R_{0,n}^{na}(D)$.
\begin{lemma}\label{$Q_ad$ w$^*$-closed}\cite{farzad06}
For ${d_{0,n}}{\in}L_1(\mu_{0,n},BC({\cal Y}_{0,n}))$, the set $Q_{0,n}(D)$ is bounded and weak$^*$-closed subset of $Q_{ad}$ (hence weak$^*$-compact).
\end{lemma}

Now we prepare to consider the problem stated in Definition~\ref{def1}. First, we show weak$^*$-compactness of $\overrightarrow{Q}_{ad}$ defined as a subset of $Q_{ad}$ as follows.
\begin{eqnarray*}
{\overrightarrow Q}_{ad}=\Big\{q_{0,n}\in {Q_{ad}}:q_{0,n}(dy^n;x^n)={\overrightarrow q}_{0,n}(dy^n;x^n)-a.s.\Big\}.
\end{eqnarray*}
The average distortion function for the nonanticipative RDF is defined by
\begin{eqnarray*}
\overrightarrow{Q}_{0,n}(D)&\triangleq& \Big\{{q}_{0,n} \in {Q}_{ad} :\ell_{d_{0,n}}({q}_{0,n})\triangleq \int_{{\cal X}_{0,n}} \biggr(\int_{{\cal Y}_{0,n}}d_{0,n}(x^n,y^n){q}_{0,n}(dy^n;x^n) \biggr)\\
&&\qquad\qquad\qquad\otimes\mu_{0,n}(dx^n)\leq D\Big\}\bigcap{\overrightarrow Q}_{ad}\\
&=&\Big\{{q}_{0,n} \in {\overrightarrow Q}_{ad} :\ell_{d_{0,n}}({q}_{0,n})\triangleq \int_{{\cal X}_{0,n}} \biggr(\int_{{\cal Y}_{0,n}}d_{0,n}(x^n,y^n){ q}_{0,n}(dy^n;x^n) \biggr)\\
&&\qquad\qquad\qquad\otimes\mu_{0,n}(dx^n)\leq{D}\Big\},~D\geq0.
\end{eqnarray*}
Since we are interested in proving existence of nonanticipative RDF of Definition~\ref{def1}, we shall first show that $\overrightarrow{Q}_{ad}$ is weak$^*$-closed, and then utilize Lemma~\ref{$Q_ad$ w$^*$-closed} to establish weak$^*$-compactness for $\overrightarrow{Q}_{ad}$ and then weak$^*$-compactness of $\overrightarrow{Q}_{0,n}(D)$.
\begin{lemma}\label{weakstar-closed}
Let ${\cal X}_{0,n}$ and ${\cal Y}_{0,n}$ be Polish spaces and introduce the net $\{q^{\alpha}_i(dy_i;y^{i-1},x^i)\}$, where $\alpha\in({\cal D},\succeq)$, and $q^{\alpha}_i\in{\cal Q}({\cal Y}_i;{\cal Y}_{0,i-1},{\cal X}_{0,i})$. Assume
\begin{description}
\item[(a)] $q^{\alpha}_i(\cdot;y^{i-1},x^i)\buildrel w^* \over {\longrightarrow}q_i^{0}(\cdot;y^{i-1},x^i)$ for $i=1,\ldots,n$;
\item[(b)] for all $h_i(\cdot,\cdot){\in}L_1(\mu_{i},BC({\cal Y}_{i}))$ the function
\begin{eqnarray}
(x^{i},y^{i-1})\in{\cal X}_{0,i}\times{\cal Y}_{0,i-1}\longmapsto\int_{{\cal X}_i}\int_{{\cal Y}_i}h_i(y)q_i(dy;y^{i-1},x^i)\mu_i(dx_i;x^{i-1})\nonumber
\end{eqnarray}
is $L_1(\mu_{0,i-1},BC({\cal Y}_{0,i-1}))$ for $i=0,1,\ldots,n$;
\item[(c)] for all $h_i(\cdot,\cdot){\in}L_1(\mu_{i},BC({\cal Y}_{i}))$ and $\forall~\epsilon>0$ there exists $\alpha\succ\alpha_{\epsilon}$ such that
\begin{eqnarray}
&&\int_{{\cal X}_i}\sup_{y^{i-1}\in{\cal Y}_{0,i-1}}\bigg{|}\int_{{\cal Y}_i}h_i(x_i,y_i)q_i^{\alpha}(dy_i;y^{i-1},x^i)\nonumber\\
&&-\int_{{\cal Y}_i}h_i(x_i,y_i)q_i^0(dy_i;y^{i-1},x^i)\bigg{|}\mu_i(dx_i;x^{i-1})<\epsilon,\quad\forall~x^{i-1}\in{\cal X}_{0,i-1}\nonumber.
\end{eqnarray}
\end{description}
Then the convolution of stochastic kernels converges in weak$^*$-sense as follows.
\begin{equation}
{\overrightarrow{q}}_{0,n}^{\alpha}\buildrel w^* \over \longrightarrow{\overrightarrow{q}}_{0,n}^{0}\label{eq.8}
\end{equation}
e.g, the set $\overrightarrow{Q}_{ad}$ is weak$^*$-closed.
\end{lemma}
\begin{proof}
See Appendix.
\end{proof}
\vspace*{0.2cm}
\noi Next, we utilize the weak$^*$-compactness of $\overrightarrow{Q}_{ad}$ to show that $\overrightarrow{Q}_{0,n}(D)$ is also weak$^*$-compact.
\begin{remark}\label{remove_boundness}
There are certain important cases in which $d_{0,n}$ may not be bounded. This is the case when $d_{0,n}$ is a metric of a linear metric space. The next theorem is crucial in showing the weak$^*$-closedness property of $\overrightarrow{Q}_{0,n}(D)$ to those distortion functions $d_{0,n}$ which are not necessarily bounded, since they are measurable functions from the class $d_{0,n}\in{L}_1(\mu_{0,n},BC({\cal Y}_{0,n}))$.
\end{remark}
\begin{theorem}\label{weakstar-compact_2}
Let  ${\cal X}_{0,n},{\cal Y}_{0,n}$ be two Polish spaces  and $d_{0,n} :{\cal X}_{0,n}\times{\cal Y}_{0,n}\mapsto[0,\infty]$, a measurable, nonnegative, extended real valued function, such that for a fixed $x^n \in {\cal X}_{0,n}$, $y^n \rightarrow d(x^n,\cdot)$ is continuous on ${\cal Y}_{0,n}$, for $\mu_{0,n}$-almost all $x^n \in {\cal X}_{0,n}$ and suppose the conditions of Lemma~\ref{weakstar-closed} hold. For any $D \in [0,\infty)$, the set ${\overrightarrow Q}_{0,n}(D)$  is a weak$^*$-closed subset of ${\overrightarrow Q}_{ad}$ and hence weak$^*$-compact.
\end{theorem}
\begin{proof}
Let $\{{\overrightarrow q}_{0,n}^{\alpha}\}\in {\overrightarrow Q}_{0,n}(D) \subset {\overrightarrow Q}_{ad}$ be a net. Since ${\overrightarrow Q}_{ad}$ is weak$^*$-compact, there exists a subnet of the net $\{{\overrightarrow q}_{0,n}^{\alpha}\},$ relabelled as the original net, and an element ${\overrightarrow q}_{0,n}^0 \in {\overrightarrow Q}_{ad}$ such that ${\overrightarrow q}_{0,n}^{\alpha} \buildrel w^* \over \longrightarrow {\overrightarrow q}_{0,n}^0$\footnote{i.e.$\Big| \int_{{\cal X}_{0,n}} \int_{{\cal Y}_{0,n}} \phi(x^n,y^n)
\overrightarrow{q}_{0,n}^{\alpha}(dy^n;x^n)\otimes\mu_{0,n}(dx^n)-\int_{{\cal X}_{0,n}}\int_{{\cal Y}_{0,n}} \phi(x^n,y^n)\overrightarrow{q}_{0,n}^{0}(dy^n;x^n)\otimes\mu_{0,n}(dx^n)\Big| \longrightarrow 0$  for any $\phi \in L_1(\mu_{0,n};BC({\cal Y}_{0,n}))$.}. We must show that ${\overrightarrow q}_{0,n}^0  \in {\overrightarrow Q}_{0,n}(D).$ Considering  the sequence  $\{d_{0,n}^k \triangleq d_{0,n}\wedge k, k \in N\},$ which are bounded, measurable functions (continuous in the second argument), it follows from the weak$^*$-convergence of the sequence $\{{\overrightarrow q}_{0,n}^{\alpha}\}$  to ${\overrightarrow q}_{0,n}^0$ that
\begin{eqnarray}
&&\left.\int_{{\cal X}_{0,n}}\bigg( \int_{{\cal Y}_{0,n}}d_{0,n}^k(x^n,y^n){\overrightarrow q}_{0,n}^{0}(dy^n;x^n)\bigg)\mu_{0,n}(dx^n) \right.\nonumber \\[-1.5ex]
\label{eq.i.5}\\[-1.5ex]
&&\quad\left.=\lim_{\alpha}\int_{{\cal X}_{0,n}} \bigg( \int_{{\cal Y}_{0,n}}d_{0,n}^k(x^n,y^n){\overrightarrow q}_{0,n}^{\alpha}(dy^n;x^n)\bigg)\mu_{0,n}(dx^n) \right.\nonumber
\end{eqnarray}
for each $k\in N$. Since $d_{0,n}$ is non-negative  and  $d_{0,n}^k\uparrow d_{0,n}$ as
$k \longrightarrow \infty$ and ${\overrightarrow q}_{0,n}^{\alpha}\in {\overrightarrow Q}_{0,n}(D)$, we have
\begin{eqnarray*}
\int_{{\cal X}_{0,n}}\bigg( \int_{{\cal Y}_{0,n}}
&d_{0,n}^k&(x^n,y^n){\overrightarrow q}_{0,n}^{0}(dy^n;x^n)\bigg)\mu_{0,n}(dx^n)\nonumber\\
&=&\lim_{\alpha}\int_{{\cal X}_{0,n}} \biggl( \int_{{\cal Y}_{0,n}}d_{0,n}^k(x^n,y^n){\overrightarrow q}_{0,n}^{\alpha}(dy^n;x^n)\biggr)\mu_{0,n}(dx^n)\nonumber\\
&\leq&\lim_{\alpha}\int_{{\cal X}_{0,n}} \bigg( \int_{{\cal Y}_{0,n}}d_{0,n}(x^n,y^n){\overrightarrow q}_{0,n}^{\alpha}(dy^n;x^n)\bigg)\mu_{0,n}(dx^n)\leq D \label{eq.i.6}
\end{eqnarray*}
which is valid for all $k \in N$. Since $d_{0,n}^k
\uparrow d_{0,n}$ and they are non-negative, it follows from Lebesgue's
monotone convergence theorem and non-negativity of stochastic kernels that
\begin{eqnarray*}
\int_{{\cal X}_{0,n}}\bigg( \int_{{\cal Y}_{0,n}}d_{0,n}(x^n,y^n){\overrightarrow q}_{0,n}^{0}(dy^n;x^n)\bigg)\mu_{0,n}(dx^n)\leq D.
\end{eqnarray*}

This shows that the weak$^*$-limit ${\overrightarrow q}_{0,n}^0 \in{\overrightarrow Q}_{0,n}(D)$ and hence we have proved  that the set ${\overrightarrow Q}_{0,n}(D)$  is a weak$^*$-closed subset of ${\overrightarrow Q}_{ad}$. By Alaoglu's theorem \cite{dunford1988} being a weak$^*$-closed subset of a weak$^*$-compact set, it is weak$^*$-compact.
\end{proof}
\vspace*{0.2cm}
\noi Based on Theorem~\ref{weakstar-compact_2} and lower semicontinuity of relative entropy, we show existence of the optimal reconstruction conditional distribution for nonanticipative RDF.
\begin{theorem} \label{th3}$(${\bf Existence}$)$
Under the conditions of Theorem~\ref{weakstar-compact_2}, $R^{na}_{0,n}(D)$ has a minimum.
\end{theorem}
\begin{proof}
This follows from Theorem~\ref{weakstar-compact_2} provided lower semicontinuity of ${\mathbb{I}}(\mu_{0,n},\cdot)$ on $\overrightarrow{Q}_{ad}$ is established. First we prove that $\overrightarrow{q}_{0,n}\rightarrow \mathbb{I}({\mu}_{0,n},\cdot)$ is weak$^*$-lower semicontinuous. Let $\{\overrightarrow{q}_{0,n}^{\alpha}\}$ be a net from $\overrightarrow{Q}_{ad}$ and suppose it is weak$^*$-convergent to ${\overrightarrow q}_{0,n}^0$. Define the net $P_{0,n}^{\alpha} \in \Pi_{rba}({\cal X}_{0,n}\times {\cal Y}_{0,n})$ given by the convolution product $P_{0,n}^{\alpha} \equiv \mu_{0,n}(dx^n) \otimes  {\overrightarrow q}_{0,n}^{\alpha}(dy^n;x^n)$. Take any $\varphi(\cdot)\in BC({\cal X}_{0,n}\times {\cal Y}_{0,n})$ and consider the expression
\begin{eqnarray}
\int_{{\cal X}_{0,n}\times{\cal Y}_{0,n}} \varphi_{0,n}(x^n,y^n) P_{0,n}^{\alpha}(dx^n,dy^n)\equiv \int_{{\cal X}_{0,n}\times{\cal Y}_{0,n}} \varphi_{0,n}(x^n,y^n) {\overrightarrow q}_{0,n}^{\alpha}(dy^n;x^n)\otimes\mu_{0,n}(dx^n)\nonumber.
\end{eqnarray}
Since ${\overrightarrow q}_{0,n}^{\alpha}\buildrel w^* \over \longrightarrow {\overrightarrow q}_{0,n}^0$
in $L_{\infty}^w(\mu_{0,n},\Pi_{rba}({\cal Y}_{0,n}))$, it is clear from the above
expression  that
\begin{eqnarray}
P_{0,n}^{\alpha} \buildrel w^* \over \longrightarrow P_{0,n}^0 \equiv \mu_{0,n} \otimes {\overrightarrow q}_{0,n}^0~~\hbox{in}~~\Pi_{rba}({\cal X}_{0,n}\times {\cal Y}_{0,n}).
\end{eqnarray}
Similarly one can easily verify that the net of the  product measures $\{\pi_{0,n}^{\alpha}\}$ converges to the product measure $\pi_{0,n}^0$,
\begin{eqnarray}
\pi_{0,n}^{\alpha} \equiv \nu_{0,n}^{\alpha}\times \mu_{0,n} \buildrel w^* \over
\longrightarrow \nu_{0,n}^0 \times \mu_{0,n}\equiv \pi_{0,n}^0\nonumber
\end{eqnarray}
where $\{\nu_{0,n}^{\alpha}\}$ are the marginals  of $\{P_{0,n}^{\alpha}\}$ on ${\cal Y}_{0,n}$ and $\nu_{0,n}^0$ is its weak$^*$-limit. Now we use the lower semicontinuity property of relative entropy \cite[Lemma~1.4.3, p.~36]{dupuis-ellis97}. Following  \cite{dupuis-ellis97} it is verified that the same procedure holds true not only for countably additive measures but also for finitely additive ones. Using this fact we conclude that
\begin{eqnarray}
\mathbb{D}(P_{0,n}||\pi_{0,n})\leq\liminf_{\alpha\longrightarrow\infty}\mathbb{D}(P_{0,n}^{\alpha}||\pi_{0,n}^{\alpha})\nonumber.
\end{eqnarray}
By (\ref{re4}), this is equivalent to
\begin{eqnarray}
{\mathbb{I}}(\mu_{0,n},q_{0,n})\leq\liminf_{\alpha\longrightarrow\infty}{\mathbb{I}}(\mu_{0,n},q_{0,n}^{\alpha}).
\end{eqnarray}
This proves weak$^*$-lower semicontinuity of ${\mathbb{I}}(\mu_{0,n},\cdot)$ on $\overrightarrow{Q}_{ad}$. We have already observed in Theorem~\ref{weakstar-compact_2} that the set $\overrightarrow{Q}_{0,n}(D)$ is weak$^*$-compact, and we have just seen that
${\mathbb{I}}(\mu_{0,n},\cdot)$  is weak$^*$-lower semicontinuous. Hence ${\mathbb{I}}(\mu_{0,n},\cdot)$ attains its infimum on $\overrightarrow{Q}_{0,n}(D)$. So there exists a $\overrightarrow{q}_{0,n}^* \in \overrightarrow{Q}_{0,n}(D)$ such that $R_{0,n}^{na}(D)={\mathbb{I}}(\mu_{0,n},\overrightarrow{q}_{0,n}^*)$.
\end{proof}

%%%%%%%%%%%%%%%%%%%%%%%%%%%%%%%%%%%%%%%%%%%%%%%%%%%%%%%%%%%%%%%%%%%%%%%%%%%%%%%%%%%%%%%%%%%

\section{Necessary Conditions of Optimality for Nonanticipative RDF}\label{necessary}

In this section the form of the optimal nonanticipative convolution reconstruction kernels is derived under a stationarity assumption. The method is based on calculus of variations on the space of measures \cite{dluenberger69}.
\begin{assumption}\label{stationarity}
The family of measures $\overrightarrow{q}_{0,n}(dy^n;x^n)$ defined in (\ref{eq.2}), is the convolution of stationary conditional distributions.
\end{assumption}

Assumption~\ref{stationarity} holds for stationary process $\{(X_i,Y_i):i\in\mathbb{N}\}$ and single letter distortion $d_{0,n}(x^n,y^n)\equiv\sum_{i=0}^n\rho(x_i,y_i)$. It also holds for distortion defined by $\rho(T^i{x^n},T^i{y^n})$, for which $T^i{x^n}=\tilde{x}^{n}$ is the $i^{th}$ shift operator on the input sequence $x^n$, where $\tilde{x}_{k}=x_{k+i}$ (similarly for $T^i{y^n}$), and $\sum_{i=0}^n\rho(T^ix^n,T^iy^n)$ depends only on the components of $(x^n,y^n)$  \cite{gray2010}. Utilizing Assumption~\ref{stationarity}, which holds for stationary processes  and a single letter distortion function, the Gateaux differential of $\mathbb{I}(\mu_{0,n},\overrightarrow{q}_{0,n})$ is taken at $\overrightarrow{q}_{0,n}^*$ in the direction of $\overrightarrow{q}_{0,n}-\overrightarrow{q}_{0,n}^*$, via the definition $\overrightarrow{q}_{0,n}^{\epsilon}\triangleq\overrightarrow{q}_{0,n}+\epsilon\big{(}
\overrightarrow{q}_{0,n}-\overrightarrow{q}_{0,n}^*\big{)}$, $\epsilon\in[0,1]$, since under the stationarity assumption, the functionals $\{q_i(dy_i;y^{i-1},x^i)\in{\cal Q}({\cal Y}_i;{\cal Y}_{0,i-1}\times{\cal X}_{0,i}):~i=0,1,\ldots,n\}$ are identical.
\begin{theorem} \label{th5}
Suppose ${\mathbb I}_{\mu_{0,n}}(\overrightarrow{q}_{0,n}) \triangleq {\mathbb I}(\mu_{0,n},\overrightarrow{q}_{0,n})$ is well defined for every $\overrightarrow{q}_{0,n}\in L_{\infty}^w(\mu_{0,n},$ $\Pi_{rba}({\cal Y}_{0,n}))$ possibly taking values from the set $[0,\infty].$ Then  $\overrightarrow{q}_{0,n} \rightarrow {\mathbb I}_{\mu_{0,n}}(\overrightarrow{q}_{0,n})$ is Gateaux differentiable at every point in $L_{\infty}^w(\mu_{0,n},\Pi_{rba}({\cal Y}_{0,n})),$  and the Gateaux derivative at the  point $\overrightarrow{q}_{0,n}^*$ in the direction $\overrightarrow{q}_{0,n}-\overrightarrow{q}_{0,n}^*$ is given
by
\begin{eqnarray}
\delta{\mathbb I}_{\mu_{0,n}}(\overrightarrow{q}_{0,n}^*,\overrightarrow{q}_{0,n}-\overrightarrow{q}_{0,n}^*)=\int_{{\cal X}_{0,n}}\int_{{\cal Y}_{0,n}}\log \Bigg(\frac{\overrightarrow{q}_{0,n}^*(dy^n;x^n)}{\nu_{0,n}^*(dy^n)}\Bigg)(\overrightarrow{q}_{0,n}-\overrightarrow{q}_{0,n}^*)(dy^n;x^n)\otimes \mu_{0,n}(dx^n)\nonumber
\end{eqnarray}
where $\nu_{0,n}^*\in{\cal M}_1({\cal Y}_{0,n})$ is the marginal measure corresponding
to $\overrightarrow{q}_{0,n}^*\otimes\mu_{0,n}\in{\cal M}_1({\cal Y}_{0,n}\times{\cal X}_{0,n})$.
\end{theorem}
\begin{proof}
The proof, although lengthy, it is similar to the one in \cite{farzad06}, hence it is omitted.
\end{proof}
\vspace*{0.2cm}
\noi The constrained problem defined by (\ref{ex12}) can be reformulated using Lagrange multipliers. The equivalence of constrained and unconstrained problems is established in the following theorem.
\begin{theorem} \label{lagrange_duality}
Suppose $d_{0,n}(x^n,y^n)\triangleq\sum_{i=0}^n\rho(T^i{x^n},T^i{y^n})$, where $d_{0,n}: {\cal X}_{0,n}\times{\cal Y}_{0,n} \rightarrow \overline R_0 \equiv [0,\infty]$ is continuous in the second argument and the set $\Gamma \equiv\{(x^n,y^n) \in {\cal X}_{0,n}\times{\cal Y}_{0,n}: d_{0,n}(x^n,y^n) < D \} $ is nonempty. Then the constrained problem as stated in Theorem~\ref{th3}, is equivalent to an unconstrained problem stated below.
\begin{align}
\inf_{\overrightarrow{q}_{0,n} \in \overrightarrow{Q}_{0,n}(D)} \mathbb{I}({\mu_{0,n}},\overrightarrow{q}_{0,n}) &= \max_{s \leq 0} \inf_{\overrightarrow{q}_{0,n}}\{ \mathbb{I}({\mu_{0,n}},\overrightarrow{q}_{0,n}) - sG(\overrightarrow{q}_{0,n})\},~G(\overrightarrow{q}_{0,n})\tri\ell_{d_{0,n}}(\overrightarrow{q}_{0,n})-D \nonumber\\
&= \max_{s \leq 0}\inf_{\overrightarrow{q}_{0,n}} \Big\{ \mathbb{I}({\mu_{0,n}},\overrightarrow{q}_{0,n}) - s \Big(\int_{{\cal X}_{0,n}} \int_{{\cal Y}_{0,n}}d_{0,n}(x^n,y^n)\nonumber\\
&\qquad{\overrightarrow q}_{0,n}(dy^n;x^n)\otimes\mu_{0,n}(dx^n) -D \Big)\Big\}\nonumber
\end{align}
where $\overrightarrow{q}_{0,n}\equiv\overrightarrow{q}_{0,n}(dy^n;x^n)=\otimes_{i=0}^n{q}_i(dy_i;y^{i-1},x^i)$-a.s. Further the infimum occurs on the boundary of the  set  $\overrightarrow {Q}_{0,n}(D)$.
\end{theorem}
\begin{proof}
See Appendix.
\end{proof}
\vspace*{0.2cm}
\noi Utilizing Theorem~\ref{lagrange_duality}, we can reformulate the constraint problem as an unconstrained problem, hence we have
\begin{equation}
{R}_{0,n}^{na}(D) = \sup_{s\leq{0}}\inf_{{\overrightarrow{q}_{0,n}}} \Big\{{{\mathbb I}}(\mu_{0,n},\overrightarrow{q}_{0,n})-s(\ell_{{d}_{0,n}}(\overrightarrow{q}_{0,n})-D)\Big\}. \label{ex13}
\end{equation}
Note that ${\overrightarrow{q}_{0,n}} \in {\cal M}_1({\cal Y}_{0,n})$ are probability measures on ${\cal Y}_{0,n}$ therefore, one should introduce another set of Lagrange multipliers.\\
Moreover, $\overrightarrow{q}_{0,n}(dy^n;x^n)=\otimes_{i=0}^n{q}_i(dy_i;y^{i-1},x^i)$ is a consistent probability measure on ${\cal Y}_{0,n}$, therefore for each $k=0,1,\ldots,n$, $\int_{{\cal Y}_{0,k}}\overrightarrow{q}_{0,k}(dy^k;x^k)=1$. This constraint is expressed via
\begin{eqnarray}
&&\left.\sum_{i=0}^n\int_{{\cal X}_{0,i}\times{\cal Y}_{0,i}}\lambda_{i}(x^i,y^{i-1})\Big{(}\overrightarrow{q}_{0,i}(dy^i;x^i)-1\Big{)}\mu_{0,i}(dx^i)\right.\nonumber\\[-1.5ex]\label{eq.6}\\[-1.5ex]
&&\quad\left.=\sum_{i=0}^n\int_{{\cal X}_{0,n}\times{\cal Y}_{0,n}}\lambda_{i}(x^i,y^{i-1})\Big{(}\overrightarrow{q}_{0,n}(dy^n;x^n)-1\Big{)}\mu_{0,n}(dx^n)\right.\nonumber
\end{eqnarray}
where $\{\lambda_{i}(\cdot,\cdot):~i=0,1,\ldots,n\}$ are Lagrange multipliers.\\
Utilizing the additional constraint (\ref{eq.6}) in (\ref{ex13}), then we derive the optimal reconstruction kernel for the nonanticipative RDF, $R_{0,n}^{na}(D)$. This is given in the following theorem.
\begin{theorem} \label{th6}
Suppose $d_{0,n}(x^n,y^n)=\sum_{i=0}^n\rho(T^i{x^n},T^i{y^n})$ and the conditions of Lemma~\ref{weakstar-closed} and Theorem~\ref{weakstar-compact_2} hold. Then the infimum in (\ref{ex13}) is attained at  $\overrightarrow{q}^*_{0,n} \in{L}_{\infty}^w(\mu_{0,n},{\Pi}_{rba}({\cal Y}_{0,n}))$ given by\footnote{Due to stationarity assumption $\nu_i(\cdot;\cdot)=\nu(\cdot;\cdot)$ and $q^*_i(\cdot;\cdot,\cdot)=q^*(\cdot;\cdot,\cdot)$,~$\forall~i=0,1,\ldots,n$.}
\begin{eqnarray}
\overrightarrow{q}^*_{0,n}(dy^n;x^n)&=&\otimes_{i=0}^nq_i^*(dy_i;y^{i-1},x^i)-a.s\nonumber\\
&=&\otimes_{i=0}^n\frac{e^{s \rho(T^i{x^n},T^i{y^n})}\nu^*_i(dy_i;y^{i-1})}{\int_{{\cal Y}_i} e^{s \rho(T^i{x^n},T^i{y^n})} \nu^*_i(dy_i;y^{i-1})},~s\leq{0}\label{ex14}
\end{eqnarray}
and $\nu^*_i(dy_i;y^{i-1})\in {\cal Q}({\cal Y}_i;{\cal Y}_{0,{i-1}})$. The nonanticipative RDF is given by 
\begin{eqnarray}
{R}_{0,n}^{na}(D)&=&sD -\sum_{i=0}^n\int_{{{\cal X}_{0,i}}\times{{\cal Y}_{0,i-1}}}\log \Big( \int_{{\cal Y}_i} e^{s\rho(T^i{x^n},T^i{y^n})} \nu^*_i(dy_i;y^{i-1})\Big)\nonumber\\
&&\quad\times{{\overrightarrow q}^*_{0,i-1}}(dy^{i-1};x^{i-1})\otimes\mu_{0,i}(dx^i)\label{ex15}
\end{eqnarray}
where ``$s$" is the optimal value of (\ref{ex13}).\\
If ${R}_{0,n}^{na}(D) > 0$ then $ s < 0$  and
\begin{eqnarray}
\sum_{i=0}^n\int_{{\cal X}_{0,i}} \int_{{\cal Y}_{0,i}}
\rho(T^i{x^n},T^i{y^n}){\overrightarrow q}^*_{0,i}(dy^i;x^i)\otimes \mu_{0,i}(dx^i)=D\label{eq.7}
\end{eqnarray}
and s is obtained from the equality condition (\ref{eq.7}).
\end{theorem}
\begin{proof}
The fully unconstrained problem of (\ref{ex13}) is obtained by introducing another set of Lagrange multipliers $\{\lambda_{i}(\cdot,\cdot):~i=0,1,\ldots,n\}$ as in (\ref{eq.6}). Using the pair of Lagrange multipliers $\{s,\lambda\triangleq\{\lambda_{i}(\cdot,\cdot):~i=0,1,\ldots,n\}\}$ introduce the extended pay-off functional
\begin{eqnarray}
\mathbb{I}^{s,\lambda}_D(\mu_{0,n},\overrightarrow{q}_{0,n})&\triangleq&\mathbb{I}(\mu_{0,n},\overrightarrow{q}_{0,n})-s\Big{(}\ell_{d_{0,n}}(\overrightarrow{q}_{0,n})-D\Big{)}\nonumber\\
&+&\sum_{i=0}^n\int_{{\cal X}_{0,n}} \int_{{\cal Y}_{0,n}}\lambda_i(x^i,y^{i-1})\Big{(}\overrightarrow{q}_{0,n}(dy^n;x^n)-1\Big{)}\mu_{0,n}(dx^n).\nonumber
\end{eqnarray}
This is a fully unconstrained problem on the vector space ${L}_{\infty}^{w}(\mu_{0,n},M_{rba}({\cal Y}_{0,n}))$. Utilizing Theorem~\ref{th5}, the Gateaux derivative of $\mathbb{I}_D^{s,\lambda}$ on ${L}_{\infty}^{w}(\mu_{0,n},M_{rba}({\cal Y}_{0,n}))$ at any point $\overrightarrow{q}_{0,n}^*$ in the direction $\overrightarrow{q}_{0,n}-\overrightarrow{q}_{0,n}^*$ is given by
\begin{align}
\delta\mathbb{I}^{s,\lambda}_D(\overrightarrow{q}^*_{0,n};\overrightarrow{q}_{0,n}-\overrightarrow{q}_{0,n}^*)&=\int_{{\cal X}_{0,n}\times{\cal Y}_{0,n}}\log\Bigg(\frac{\overrightarrow{q}_{0,n}^*(dy^n;x^n)}{\nu_{0,n}^*(dy^n)}\Bigg)(\overrightarrow{q}_{0,n}-\overrightarrow{q}_{0,n}^*)(dy^n;x^n)\otimes\mu_{0,n}(dx^n)\nonumber\\
&\quad-s\int_{{\cal X}_{0,n}\times{\cal Y}_{0,n}}d_{0,n}(x^n,y^n)(\overrightarrow{q}_{0,n}-\overrightarrow{q}_{0,n}^*)(dy^n;x^n)\otimes\mu_{0,n}(dx^n)\nonumber\\
&\quad+\sum_{i=0}^n\int_{{\cal X}_{0,n}\times{\cal Y}_{0,n}}\lambda_i(x^i,y^{i-1})(\overrightarrow{q}_{0,n}-\overrightarrow{q}_{0,n}^*)(dy^n;x^n)\otimes\mu_{0,n}(dx^n)\nonumber
\end{align}
\begin{align}
&=\int_{{\cal X}_{0,n}\times{\cal Y}_{0,n}}\log\Bigg(e^{\sum_{i=0}^n\big{(}-s\rho(T^i{x^n},T^i{y^n})+\lambda_i(x^i,y^{i-1})\big{)}}
\frac{\overrightarrow{q}_{0,n}^*(dy^n;x^n)}{\nu_{0,n}^*(dy^n)}\Bigg)\nonumber\\
&\qquad(\overrightarrow{q}_{0,n}-\overrightarrow{q}_{0,n}^*)(dy^n;x^n)\otimes\mu_{0,n}(dx^n),~\forall\overrightarrow{q}_{0,n}\in{L}_{\infty}^{w}(\mu_{0,n},M_{rba}({\cal Y}_{0,n}))\nonumber.
\end{align}
Since $\mathbb{I}^{s,\lambda}_D(\mu_{0,n},\overrightarrow{q}_{0,n})$ is convex in $\overrightarrow{q}_{0,n}$, it follows from the calculus of variations principle that a necessary and sufficient condition for $\overrightarrow{q}^*_{0,n}$ to be a minimizer is $\delta\mathbb{I}^{s,\lambda}_D(\overrightarrow{q}^*_{0,n};\overrightarrow{q}_{0,n}-\overrightarrow{q}_{0,n}^*)=0$, $\forall\overrightarrow{q}_{0,n}\in{L}_{\infty}^{w}(\mu_{0,n},M_{rba}({\cal Y}_{0,n}))$. Since the Gateaux derivative must be zero for all $\overrightarrow{q}_{0,n}\in{L}_{\infty}^{w}(\mu_{0,n},M_{rba}({\cal Y}_{0,n}))$ then
\begin{eqnarray}
\frac{\overrightarrow{q}_{0,n}^*(dy^n;x^n)}{\nu_{0,n}^*(dy^n)}=
e^{\sum_{i=0}^n\big{(}s\rho(T^i{x^n},T^i{y^n})-\lambda_i(x^i,y^{i-1})\big{)}}-a.s.\nonumber
\end{eqnarray}
Equivalently,
\begin{eqnarray}
\otimes_{i=0}^n\frac{q_i^*(dy_i;y^{i-1},x^i)}{\nu_{i}^*(dy_i;y^{i-1})}
=\otimes_{i=0}^n{e}^{\big{(}s\rho(T^i{x^n},T^i{y^n})-\lambda_i(x^i,y^{i-1})\big{)}}-a.s.\nonumber
\end{eqnarray}
Since $\int_{{\cal Y}_i}q_i^*(dy_i;y^{i-1},x^i)=1$, then
\begin{eqnarray}
\lambda_i(x^i,y^{i-1})=\log\int_{{\cal Y}_i}e^{s\rho(T^i{x^n},T^i{y^n})}\nu_i^*(dy_i;y^{i-1}),~i=0,1,\ldots,n\nonumber.
\end{eqnarray}
Hence,
\begin{eqnarray}
\overrightarrow{q}^*_{0,n}(dy^n;x^n)&=&\otimes_{i=0}^nq_i^*(dy_i;y^{i-1},x^i)-a.s\nonumber\\
&=&\otimes_{i=0}^n\frac{e^{s \rho(T^i{x^n},T^i{y^n})}\nu^*_i(dy_i;y^{i-1})}{\int_{{\cal Y}_i} e^{s \rho(T^i{x^n},T^i{y^n})} \nu^*_i(dy_i;y^{i-1})}.\nonumber
\end{eqnarray}
Since $s\leq0$ and $\lambda_i\geq0$,~$i=0,1,\ldots,n$ then $\overrightarrow{q}_{0,n}^*\in{L}_{\infty}^w(\mu_{0,n},{\Pi}_{rba}({\cal Y}_{0,n}))$. Substituting $\overrightarrow{q}_{0,n}^*$ into $\mathbb{I}_D^{s,\lambda}(\mu_{0,n},\overrightarrow{q}_{0,n})$ gives (\ref{ex15}).\\
Note that for $s=0$ then $R_{0,n}^{na}(D)=0$ and $\overrightarrow{q}_{0,n}^*(dy^n;x^n)=\nu_{0,n}^*(dy^n)$,~$\mu_{0,n}-$almost all $x^n\in{\cal X}_{0,n}$. This is trivial so we must have $s<0$. From Theorem~\ref{lagrange_duality} the solution occurs on the boundary of $\overrightarrow{Q}_{0,n}(D)$ giving (\ref{eq.7}) for $s<0$.
\end{proof}
\vspace*{0.2cm}
\noi Often it is interesting to identify conditions so that the optimal reconstruction is Markov with respect to $\{X_i:~i=0,1,\ldots,n\}$. The next remark discusses this case.
\begin{remark}
Note that if the distortion function satisfies $\rho(T^i{x^n},T^i{y^n})=\rho(x_i,T^i{y^n})$ then according to Theorem~\ref{th6} we have
\begin{equation}
{q}^*_{i}(dy_i;y^{i-1},x^i)=q_i^*(dy_i;y^{i-1},x_i)-a.s.,~i\in{\mathbb{N}^n}
\end{equation}
that is, the reconstruction kernel is Markov in $X^n$. However, even if $\rho(T^ix^n,T^iy^n)=\rho(x_i,y_i)$ (single letter) one cannot claim that the optimal reconstruction distribution is also Markov with respect to $\{Y_i:~i=0,1,\ldots,n\}$ because the right hand side (RHS) of (\ref{ex14}) does not satisfy $\nu_i(dy_i;y^{i-1})=\nu_i(dy_i;y_{i-1})$.
\end{remark}

The relation between nonanticipative RDF and filtering theory is developed for fixed source distribution. In the next remark we discuss extensions of the nonanticipative RDF for a class of sources and relations to robust filtering.
\begin{remark}
Nonanticipative RDF can be generalized to a class of sources to address robustness of the filter. One such class is defined by a relative entropy constraint between the unknown or true distribution $P_{X^n}$ with respect to the nominal distribution $P^0_{X^n}$ via 
\begin{align*}
{\cal M}_{P^0_{X^n}}(d)\triangleq\big\{P_{X^n}\in{\cal M}_1({\cal X}_{0,n}):\mathbb{D}(P_{X^n}||P^0_{X^n})\leq{d}\big\}
\end{align*} 
where $d$ is the radius of uncertainty. Such a model of uncertainty or  class of distributions is often employed in filtering and control applications because it is related to robust filtering and control using minimax methods \cite{xie-ugrinovskii-petersen2005,charalambous-rezaei2007}.\\
Therefore, the nonanticipative RDF for the class of sources ${\cal M}_{P^0_{X^n}}(d)$ is now defined using minimax strategies by
\begin{align}
R^{na,+}_{0,n}(D,d)=\inf_{\overrightarrow{P}_{Y^n|X^n}\in\overrightarrow{Q}_{0,n}(D)}\sup_{P_{X^n}\in{\cal M}_{P^0_{X^n}}(d)}\mathbb{I}(P_{X^n},\overrightarrow{P}_{Y^n|X^n})\label{equation111111}
\end{align}
\noi Through (\ref{equation111111}) one can obtain relations to minimax filtering strategies via nonanticipative RDF. An example using this formulation for control of Gaussian state space systems over limited rate channels is found in \cite{farhadi-charalambous2010}. The investigation of the classical RDF for such a relative entropy class of soures  is discussed in \cite{rezaei-charalambous-stavrou2010}, where it is also shown that the Von-Neumann minimax theorem holds and hence one can interchange infimum and supremum operations. The validity of the Von-Neumann minimax theorem  for (\ref{equation111111}), will imply that the optimal reconstruction distribution for the minimax nonanticipative RDF is (\ref{ex14}), and hence the remaining task is to perform the infimum operation over the relative entropy class of the solution to the nonanticipative RDF given by (\ref{ex15}). This is the simplest approach to relate nonanticipative RDF for a class of sources  and  minimax filtering techniques. Unlike minimax filtering techniques, the filtering obtained from (\ref{equation111111}) will always satisfy the fidelity criterion which can be defined with respect to probability of error or the average error.\\
\noi However, it is not clear how one can apply sensitivity minimization to nonanticipative RDF filter, because only the source distribution is given, while the observation map and filter are obtained from the realization of the optimal reconstruction distribution (see Fig.~\ref{filtering_and_causal}). This is contrary to sensitivity minimization approach, where the input-output maps are given and depend on design functions, such as, the controller or the filter \cite{hassibi-sayed-kailath1999,zhou2010}. Nevertheless, when the source is a second order Gaussian process described by a Power Spectral Density (PSD) and the fidelity of reconstruction is the mean-square error, then it might be possible to apply robust filtering and control techniques to address uncertainty of the PSD similar to the computation of capacity of channels with memory \cite{denic-charalambous-djouadi2009}.
\end{remark}
%%%%%%%%%%%%%%%%%%%%%%%%%%%%%%%%%%%%%%%%%%%%%%%%%%%%%%%%%%%%%%%%%%%%%%%%%%%%%%%%%%%%%%%%

\section{Realization of Nonanticipative RDF}\label{realization}

The realization of the nonanticipative RDF (optimal reconstruction kernel and nonanticipative RDF) is equivalent to identifying the sensor mapping (see Fig.~\ref{filtering_and_causal}) which generates the auxiliary random process $\{Z_i:~i=0,1,\ldots,n\}$ so that the optimal reconstruction conditional distribution is matched from the output of the source to the output of the filter. This intermediate mapping consists of an encoder followed by a channel. Thus, the realization of the nonanticipative optimal reconstruction distribution consists of a communication channel, an encoder and a decoder such that the reconstruction from the sequence $X^n$ to the sequence $Y^n$ matches the nonanticipative rate distortion minimizing reconstruction kernel. Fig.~\ref{realization2} illustrates a cascade of subsystems that realizes the nonanticipative RDF. For the single letter expression of classical RDF this is related to the so-called source-channel matching of information theory \cite{gastpar2003}. It is also described in \cite{charalambous2008} and \cite{tatikonda2000} for control over finite capacity communication channels, since this technique allows one to design encoding/decoding schemes without encoding and decoding delays. The realization of the optimal reconstruction kernel is given below.
\begin{definition}\label{realization1}
Given a source $\{P_{X_i|X^{i-1}}(dx_i|x^{i-1}):i=0,\ldots,n\}$,  a channel $\{P_{B_i|B^{i-1},A^{i}}(db_i|$ $b^{i-1},a^i):i=0,\ldots,n\}$ is a realization of the optimal nonanticipative reconstruction kernel $\{q_i^*(dy_i;y^{i-1},x^i):i=0,\ldots,n\}$ if there exists a pre-channel encoder $\{P_{A_i|A^{i-1},B^{i-1},X^i}(da_i|a^{i-1},$ $b^{i-1},x^i):i=0,\ldots,n\}$ and a post-channel decoder $\{P_{Y_i|Y^{i-1},B^i}(dy_i|y^{i-1},b^i):i=0,\ldots,n\}$ such that
\begin{eqnarray}
{\overrightarrow q}_{0,n}^*(dy^n;x^n) &\triangleq&\otimes_{i=0}^n q_i^*(dy_i;y^{i-1},x^i)\nonumber\\
&=&\otimes_{i=0}^nP_{Y_i|Y^{i-1},X^i}(dy_i|y^{i-1},x^i)-a.s.
\end{eqnarray}
where the joint distribution is
\begin{align}
&P_{X^n,A^n, B^n, Y^n}(dx^n,da^n,db^n,dy^n)\\
&=\otimes_{i=0}^n{P}_{Y_i|Y^{i-1},B^i,A^i,X^i}(dy_i|y^{i-1},b^i,a^i,x^i)\nonumber\\
&\quad\otimes{P}_{B_i|B^{i-1},A^i,X^i,Y^{i-1}}(db_i|b^{i-1},a^i,x^i,y^{i-1})\otimes{P}_{A_i|A^{i-1},X^i,Y^{i-1},B^{i-1}}(da_i|a^{i-1},x^i,y^{i-1},b^{i-1})\nonumber\\
&\quad\otimes{P}_{X_i|X^{i-1},A^{i-1},B^{i-1},Y^{i-1}}(dx_i|x^{i-1},a^{i-1},b^{i-1},y^{i-1})-a.s.,\nonumber\\
&=\otimes_{i=0}^n P_{Y_i|Y^{i-1},B^i}(dy_i|y^{i-1},b^i)\otimes P_{B_i|B^{i-1},A^{i}}
(db_i|b^{i-1},a^i) \nonumber \\
&\quad\otimes P_{A_i|A^{i-1},B^{i-1},X^i}(da_i|a^{i-1},b^{i-1},x^i)\otimes P_{X_i|X^{i-1}}(dx_i|x^{i-1})-a.s. \nonumber
\end{align}
The filter is given by $\{P_{X_i|B^{i-1}}(dx_i|b^{i-1}):i=0,\ldots,n\}$.
\end{definition}
\begin{figure}[ht]
\centering
\includegraphics[scale=0.70]{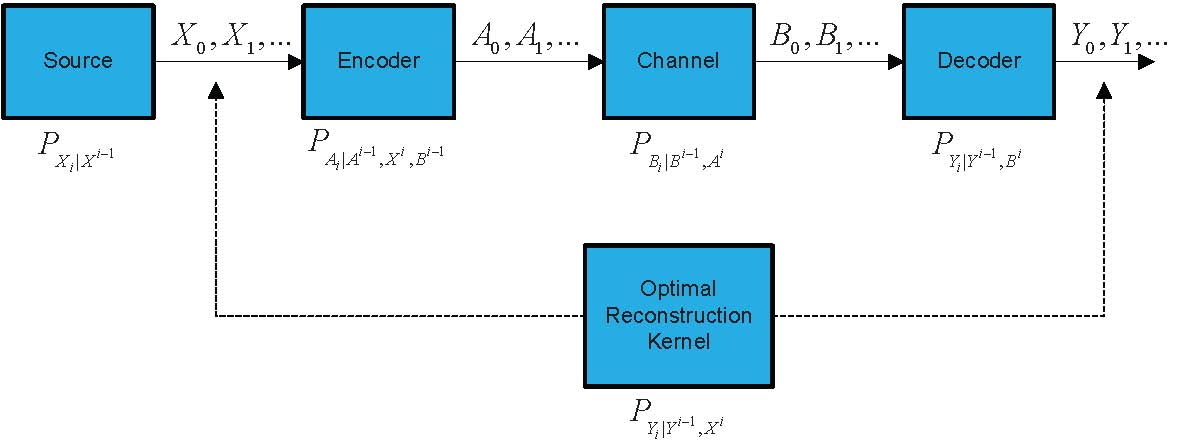}
\caption{Block Diagram of Realizable Nonanticipative Rate Distortion Function}
\label{realization2}
\end{figure}

\noi Thus, $\{B_i:~i=0,1,\ldots,n\}$ is the auxiliary random process which is obtained during the realization procedure in order to define the filter $\{P_{X_i|B^{i-1}}(dx_i|b^{i-1}):i=0,\ldots,n\}$. Note that unlike Bayesian filtering in which the auxiliary process represents the observations which are given \'a priori, in nonanticipative RDF this is identified during the realization procedure. In the Definition~\ref{realization1}, the following MC assumptions are assumed.
\begin{description}
\item[1)] $(X^i,A^i)\leftrightarrow(Y^{i-1},B^i)\leftrightarrow{Y_i}$;
\item[2)] $(X^i,Y^{i-1})\leftrightarrow(B^{i-1},A^i)\leftrightarrow{B_i}$;
\item[3)] $Y^{i-1}\leftrightarrow(A^{i-1},B^{i-1},X^i)\leftrightarrow{A_i}$;
\item[4)] $(A^{i-1},B^{i-1},Y^{i-1})\leftrightarrow{X}^{i-1}\leftrightarrow{X_i}$.
\end{description}
These conditional independent assumptions are natural since they correspond to data processing inequalities \cite{cover-thomas}. Thus, if $\{P_{B_i|B^{i-1},A^{i}}(db_i|b^{i-1},a^i):i=0,\ldots,n\}$ is a realization of the nonanticipative RDF minimizing kernel $\{q_i^*(dy_i;y^{i-1},x^i):i=0,\ldots,n\}$ then the channel connecting the source, encoder, channel, decoder achieves the nonanticipative RDF, and the filter is obtained via $\{P_{X_i|B^{i-1}}(dx_i|b^{i-1}):i=0,\ldots,n\}$. Moreover, the above MCs imply the following data processing inequality, $I (A^n \rightarrow B^n) \tri \sum_{i=0}^n{I}(A^i;B_i|B^{i-1})\geq{I}(X^n;Y^n)$. The optimal realization (encoder-channel-decoder) is defined as the one for which the last inequality holds with equality.

%%%%%%%%%%%%%%%%%%%%%%%%%%%%%%%%%%%%%%%%%%%%%%%%%%%%%%%%%%%%%%%%%%%%%%%%%%%%%%%%%%%%%%%%%%%

\section{Example}\label{example}

\par In this section, we present the filter for Gaussian Markov partially-observable processes by utilizing the realization procedure of Section~\ref{realization}.\\
Consider the following discrete-time partially observed linear Gauss-Markov system described by
\begin{eqnarray}
\left\{ \begin{array}{ll} X_{t+1}=AX_t+BW_t,~X_0=X,~t\in\mathbb{N}^n\\
Y_t=CX_t+NV_t,~t\in\mathbb{N}^n \end{array} \right.\label{equation51}
\end{eqnarray}
\begin{figure}[ht]
\centering
\includegraphics[scale=0.70]{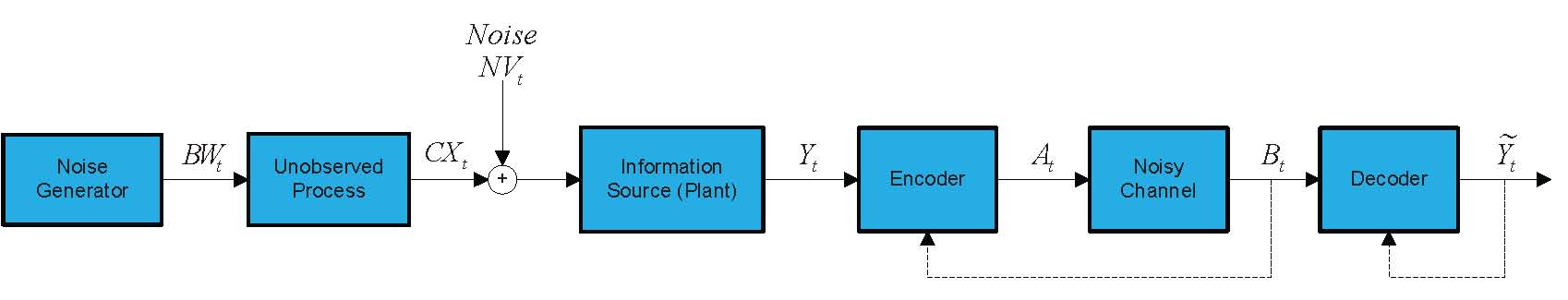}
\caption{Communication System}
\label{communication_system}
\end{figure}
where $X_t\in\mathbb{R}^m$ is the state (unobserved) process of information source (plant), and $Y_t\in\mathbb{R}^p$ is the partially observed (measurement) process. The model in (\ref{equation51}) consists of a process $\{X_t:~t\in\mathbb{N}^n\}$ which is not directly observed; instead what is directly observed is the process $\{Y_t:~t\in\mathbb{N}^n\}$ which is a noisy version of $\{X_t:~t\in\mathbb{N}^n\}$. This is a realistic model for any sensor which collects information for an underlying process, since the sensor is a measurement device which is often subject to additive Gaussian noise. Hence, in this application the objective is to compress the sensor data. Since we only treat the stationary case,  we assume that ($C,A$) is detectable and ($A,\sqrt{BB^{tr}}$) is stabilizable, ($N\neq0$) \cite{caines1988}. The state and observation noise $\{(W_t,V_t):t\in\mathbb{N}^n\}$ are mutually independent, independent of the Gaussian RV $X_0$, with parameters $N(\bar{x}_0,\bar{\Sigma}_0)$, where $W_t\in\mathbb{R}^k$ and $V_t\in\mathbb{R}^d$, are Gaussian IID processes with zero mean and identity covariances.\\
The realization will be done following Fig.~\ref{communication_system}. The goal is to reconstruct $\{Y_t:~t\in\mathbb{N}^n\}$ by $\{\tilde{Y}_t:~t\in\mathbb{N}^n\}$ causally. The distortion is single letter defined by
\begin{eqnarray*}
d_{0,n}(y^n,\tilde{y}^n)\triangleq\frac{1}{n+1}\sum_{t=0}^n||y_t-\tilde{y}_t||^2.
\end{eqnarray*}
The objective is to compute
\begin{eqnarray}
R_{0,n}^{na}(D)=\inf_{\overrightarrow{P}_{\tilde{Y}^n|Y^n}\in\overrightarrow{Q}_{0,n}(D)}\frac{1}{n+1}\mathbb{I}(P_{Y^n},\overrightarrow{P}_{\tilde{Y}^n|Y^n})\label{equation.1}
\end{eqnarray}
where $\overrightarrow{Q}_{0,n}(D)\triangleq\big{\{}\overrightarrow{P}_{\tilde{Y}^n|Y^n}:~E\{d_{0,n}(Y^n,\tilde{Y}^n)\}\leq{}D\big{\}}$, and realize the reconstruction distribution. The reconstruction of $\{X_t:~t\in\mathbb{N}^n\}$ when it is fully observed, i.e., when $Y_t=X_t$, is realized over a scalar additive white Gaussian noise (AWGN) channel in \cite{tatikonda-mitter2004}, while the partially observed scalar reconstruction of $\{Y_t:~t\in\mathbb{N}^n\}$ is realized over a scalar AWGN channel in \cite{charalambous2008} via indirect methods (utilizing upper bounds which are achievable).
\par Here, we consider the vector process $Y_t\in\mathbb{R}^p$ and realize it over a vector AWGN channel. The methodology is based on the explicit formulae of optimal reconstruction of Theorem~\ref{th6}. According to Theorem~\ref{th6}, the optimal reconstruction is given by
\begin{eqnarray}
\overrightarrow{P}^*_{\tilde{Y}^n|Y^n}(d\tilde{y}^n|y^n)=\otimes_{t=0}^n\frac{e^{s||\tilde{y}_t-y_t||^2}P_{\tilde{Y}_t|\tilde{Y}^{t-1}}(d\tilde{y}_t|\tilde{y}^{t-1})}{\int_{{\cal Y}_t}e^{s||\tilde{y}_t-y_t||^2}P_{\tilde{Y}_t|\tilde{Y}^{t-1}}(d\tilde{y}_t|\tilde{y}^{t-1})},~s\leq{0}\label{eq.9}
\end{eqnarray}
where each term in the RHS is identical because our results are derived based on the  stationarity assumption. Hence, from (\ref{eq.9}) it follows that $P_{\tilde{Y}_t|\tilde{Y}^{t-1},Y^t}=P_{\tilde{Y}_t|\tilde{Y}^{t-1},Y_t}(d\tilde{y}_t|\tilde{y}^{t-1},y_t)-$a.s., that is the reconstruction is Markov with respect to the process $\{Y_t:~t\in\mathbb{N}^n\}$. Moreover, since the exponential term $||\tilde{y}_t-y_t||^2$ in the RHS of (\ref{eq.9}) is quadratic in $(y_t,\tilde{y}_t)$, and $\{X_t:~t\in\mathbb{N}^n\}$ is Gaussian then $\{(X_t,{Y}_t):~t\in\mathbb{N}^n\}$ is jointly Gaussian, and it follows that a Gaussian distribution $P_{\tilde{Y}_t|\tilde{Y}^{t-1},Y_t}(\cdot|\tilde{y}^{t-1},y_t)$ (for a fixed realization of $(\tilde{y}^{t-1},y_t)$), and Gaussian distribution $P_{\tilde{Y}_t|\tilde{Y}^{t-1}}(\cdot|\tilde{y}^{t-1})$ can match the left and right side of (\ref{eq.9}). Therefore, at time $t\in\mathbb{N}^n$, the output $\tilde{Y}_t$ of the optimal reconstruction channel depends on $Y_t$ and the previous channel outputs $\tilde{Y}^{t-1}$, and its conditional distribution is Gaussian. Hence, the channel connecting $\{Y_t:t\in\mathbb{N}^n\}$ to $\{\tilde{Y}_t:t\in\mathbb{N}^n\}$ has the general form
\begin{eqnarray}
\tilde{Y}_t=\bar{A}_tY_t+\bar{B}_t\tilde{Y}^{t-1}+Z_t,~t\in\mathbb{N}^n\label{eq.10}
\end{eqnarray}
where $\bar{A}_t\in\mathbb{R}^{p\times{p}}$, $\bar{B}_t\in\mathbb{R}^{p\times{t}p}$, and $\{Z_t:~t\in\mathbb{N}^n\}$ is an independent sequence of Gaussian vectors. Since we treat the stationary case, the finite horizon analysis below is only an intermediate state before we give the stationary solution.   \\
The communication channel (\ref{eq.10}) can be realized via a memoryless additive Gaussian noise channel with feedback \cite{cover-thomas} defined by
\begin{eqnarray}
B_t=A_t+Z_t,~t\in\mathbb{N}^n\label{eq.11}
\end{eqnarray}
where the encoder, at time $t$, is a mapping $A_t=\Phi_t(Y_t,\tilde{Y}^{t-1})$ with power $P_t\triangleq{Trace}E\{A_tA_t^{tr}\}$, and the decoder at time $t\in\mathbb{N}^n$ receives $B^t$ and computes the reconstruction $\tilde{Y}_t=\Psi_t(B^t,\tilde{Y}^{t-1})$. By Section~\ref{realization}, in view of the MCs we have the data processing inequality $\mathbb{I}(P_{Y^n},\overrightarrow{P}_{\tilde{Y}^n|Y^n})\leq I(A^n \rightarrow B^n) = \ {I}(A^n;B^n) $, where the last equality holds because the channel is memoryless \cite{cover-thomas}.\\ %For $A_t$ Gaussian the mutual information is \cite{cover-thomas} $I(A^t;B^t)=\log|I+E\{A_tA_t^{tr}\}Cov(Z_t)^{-1}|$. \\
\noi For the realization, the first step is the whitening of the source $\{Y_t:t\in\mathbb{N}^n\}$ by introducing the Gaussian innovation process $\{K_t:~t\in\mathbb{N}^n\}$, defined by
\begin{eqnarray}
K_t\triangleq{Y}_t-E\Big{\{}Y_t|\sigma\{\tilde{Y}^{t-1}\}\Big{\}},~t\in\mathbb{N}^n\label{equation52}
\end{eqnarray}
whose covariance is defined by 
\begin{align}
\Lambda_t\triangleq{E}\{K_tK_t^{tr}\},~t\in\mathbb{N}^n.
\end{align} 
\noi The second step is the diagonalization of the covariance $\{\Lambda_t:t\in\mathbb{N}^n\}$ by introducing a unitary transformation $\{E_t:t\in\mathbb{N}^n\}$ such that 
\begin{eqnarray}
E_t\Lambda_t{E}_t^{tr}=diag\{\lambda_{t,1},\ldots\lambda_{t,p}\},~t\in\mathbb{N}^n.\label{equation53}
\end{eqnarray}
\noi Thus, $\Gamma_t\triangleq{E}_tK_t$, where $\{\Gamma_t:~t\in\mathbb{N}^{n}\}$ has independent components for each $t\in\mathbb{N}^n$. In practise, the encoder consists of a pre-encoder which preprocesses the observations $\{Y_t:t\in\mathbb{N}^n\}$ by generating $\{K_t:t\in\mathbb{N}^n\}$ and then applies $\{E_t:t\in\mathbb{N}^n\}$ to it. At the decoder end, there is a pre-decoder which generates $\{\tilde{K}_t:~t\in\mathbb{N}^n\}$ defined by
\begin{eqnarray}
\tilde{K}_t\triangleq\tilde{Y}_t-E\Big{\{}Y_t|\sigma\{\tilde{Y}^{t-1}\}\Big{\}},~t\in\mathbb{N}^n\label{eq.12}
\end{eqnarray}
\noi on which the unitary transformation $\{E_t:t\in\mathbb{N}^n\}$ is applied to generate $\tilde{\Gamma}_t=E_t\tilde{K}_t$. Next, we calculate the RDF by taking advantage of the preprocessing at the encoder-decoder. Note that the fidelity criterion $d_{0,n}(\cdot,\cdot)$ is not affected by the preprocessing at the encoder-decoder since $d_{0,n}(Y^n,\tilde{Y}^n)=d_{0,n}(K^n,\tilde{K}^n)=\frac{1}{n+1}\sum_{t=0}^n||\tilde{K}_t-K_t||^2=\frac{1}{n+1}\sum_{t=0}^n||\tilde{\Gamma}_t-\Gamma_t||^2$. Now, we show that 
\begin{align*}
\mathbb{I}(P_{Y^n},\overrightarrow{P}_{\tilde{Y}^n|Y^n})&=\sum_{t=0}^n\Big{(}H(\tilde{K}_t|\tilde{K}^{t-1})-H(\tilde{K}_t|\tilde{K}^{t-1},K_t)\Big{)}\\
&=\sum_{t=0}^n\Big{(}H(\tilde{\Gamma}_t|\tilde{\Gamma}^{t-1})-H(\tilde{\Gamma}_t|\tilde{\Gamma}^{t-1},\Gamma_t)\Big{)}.
\end{align*}
\noi By (\ref{eq.9}), 
\begin{align*}
\overrightarrow{P}_{\tilde{Y}^n|Y^n}(d\tilde{y}^n|x^n)=\otimes_{t=0}^n{P}_{\tilde{Y}_t|\tilde{Y}^{t-1},Y_t}(d\tilde{y}_t|\tilde{y}^{t-1},y_t)-a.s.
\end{align*}
Hence, 
\begin{align*}
\mathbb{I}(P_{Y^n},\overrightarrow{P}_{\tilde{Y}^n|Y^n})=\sum_{t=0}^n\Big{(}H(\tilde{Y}_t|\tilde{Y}^{t-1})-H(\tilde{Y}_t|\tilde{Y}^{t-1},Y_t)\Big{)}.
\end{align*}
Since conditional entropy is translation invariant, utilizing (\ref{eq.12}) gives 
\begin{align}
H(\tilde{Y}_t|\tilde{Y}^{t-1})&=H(\tilde{K}_t|\tilde{Y}^{t-1})\nonumber\\
&=H(\tilde{K}_t|\tilde{Y}_{-1},\tilde{Y}_0,\ldots,\tilde{Y}_{t-1})\nonumber\\
&\stackrel{(a)}=H\big{(}\tilde{K}_t|\tilde{Y}_{-1},\tilde{Y}_0,\ldots,\tilde{Y}_{t-2},\tilde{K}_{t-1}+E(\tilde{Y}_{t-1}|\sigma\{\tilde{Y}^{t-2}\})\big{)}\nonumber\\
&=H(\tilde{K}_t|\tilde{Y}_{-1},\tilde{Y}_0,\ldots,\tilde{Y}_{t-2},\tilde{K}_{t-1})\nonumber\\
&=H(\tilde{K}_t|\tilde{K}^{t-1})\label{equation61}
\end{align}
and repeated application of step $(a)$ gives (\ref{equation61}). Similarly, $H(\tilde{Y}_t|\tilde{Y}^{t-1},Y_t)=H(\tilde{K}_t|\tilde{Y}^{t-1},K_t)=H(\tilde{K}_t|\tilde{K}^{t-1},K_t)$. Hence, 
\begin{align*}
\mathbb{I}(P_{Y^n},\overrightarrow{P}_{\tilde{Y}^n|Y^n})=\sum_{t=0}^n\Big{(}H(\tilde{K}_t|\tilde{K}^{t-1})-H(\tilde{K}_t|\tilde{K}^{t-1},K_t)\Big{)}\equiv\sum_{t=0}^n{I}(K_t;\tilde{K}_t|\tilde{K}^{t-1}).
\end{align*}
\noi Since the unitary transformation is non-singular then $\mathbb{I}(P_{Y^n},\overrightarrow{P}_{\tilde{Y}^n|Y^n})=\sum_{t=0}^n{I}(K_t;\tilde{K}_t|\tilde{K}^{t-1})=\sum_{t=0}^n{I}(\Gamma_t;\tilde{\Gamma}_t|\tilde{\Gamma}^{t-1}),~{t}\in\mathbb{N}^n$. Therefore, (\ref{equation.1}) is equivalent to the following expression.
\begin{eqnarray}
R^{na}_{0,n}(D)=R_{0,n}^{na,\Gamma^n,\tilde{\Gamma}^n}(D)\triangleq\inf_{\overrightarrow{P}_{\tilde{\Gamma}^n|\Gamma^n}:~E\big{\{}d_{0,n}(\Gamma^n,\tilde{\Gamma}^n)\leq{D}\big{\}}}\frac{1}{n+1}\mathbb{I}(P_{\Gamma^n},\overrightarrow{P}_{\tilde{\Gamma}^n|\Gamma^n}).\label{equation.2}
\end{eqnarray}
\noi By \cite{berger} (invoking an upper bound and Shannon's lower bound if necessary) the stationary solution of (\ref{equation.2}) is given by
\begin{align*}
\lim_{n \longrightarrow \infty} R_{0,n}^{na,\Gamma^n,\tilde{\Gamma}^n}(D)=\lim_{n \longrightarrow \infty}\frac{1}{2}\frac{1}{n+1}\sum_{t=0}^n\sum_{i=1}^p\log\Big{(}\frac{\lambda_{t,i}}{\delta_{t,i}}\Big{)}
\end{align*}
where
\begin{eqnarray}
\delta_{t,i} \triangleq \left\{ \begin{array}{ll} \xi_t & \mbox{if} \quad \xi_t\leq\lambda_{t,i} \\
\lambda_{t,i} &  \mbox{if}\quad\xi_t>\lambda_{t,i} \end{array} \right.,~t\in\mathbb{N}^n,~i=2,\ldots,p\nonumber
\end{eqnarray}
and $\{\xi_t:~t\in\mathbb{N}^n\}$ is chosen such that $\sum_{i=1}^p\delta_{t,i}=D$. Define $\eta_{t,i}\tri1-\frac{\delta_{t,i}}{\lambda_{t,i}}$,~$i=1,\ldots,p$, $\Delta_t\triangleq{diag}\{\delta_{t,1},\ldots,\delta_{t,p}\}$, and $H_t\triangleq{diag}\{\eta_{t,1},\ldots,\eta_{t,p}\}\in\mathbb{R}^{p\times{p}}$.\\
As a result, the reconstruction conditional distribution is given by
\begin{eqnarray}
P^*_{\tilde{\Gamma}^n|\Gamma^n}(d\tilde{\gamma}^n|{\gamma}^n)=\otimes_{t=0}^n{P}^*_{\Gamma_t|\tilde{\Gamma}_t}(d\tilde{\gamma}_t|\gamma_t)-a.s.\nonumber
\end{eqnarray}
where ${P}^*_{\tilde{\Gamma}_t|{\Gamma}_t}(\cdot|\cdot)\sim{N}(H_{t}\Gamma_t,H_{t}\Delta_{t})$.\\
 %The complete design is illustrated in Fig.~\ref{discrete_time_communication_system}.
\begin{figure}[ht]
\centering
\includegraphics[scale=0.70]{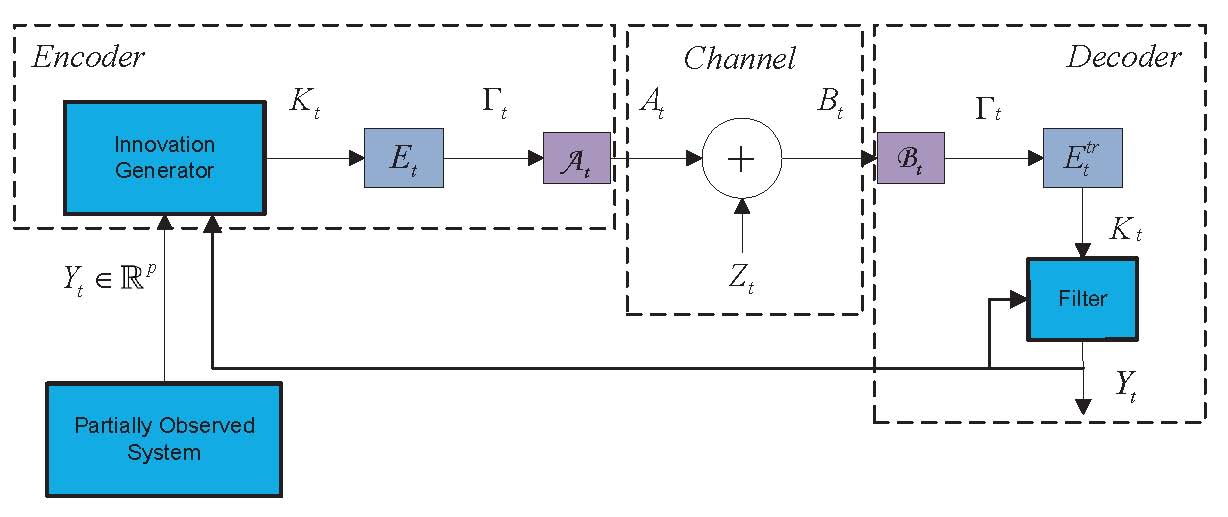}
\caption{Design of Discrete-Time Communication System}
\label{discrete_time_communication_system}
\end{figure}
\noi{\it Realization of Nonanticipative RDF Over Vector AWGN Channel.} Consider a vector channel $B_t=A_t+Z_t,~t\in\mathbb{N}^n$, where $Z_t$ is Gaussian zero mean, $Q\triangleq{C}ov(Z_t)=diag\{q_1,q_2,\ldots,q_p\}$, and $A_t\in\mathbb{R}^p$. By Section~\ref{realization}, and the memoryless nature of the channel we know that $I(A^n \rightarrow B^n) ={I}(A^n;B^n)\geq\mathbb{I}(P_{Y^n},\overrightarrow{P}_{\tilde{Y}^n|Y^n})$. Hence, we compress the source and transmit it to the decoder over the vector channel, so that the RDF is equal to the capacity of the channel, i.e., $\lim_{n \longrightarrow \infty} R_{0,n}^{na}(D)= \lim_{n \longrightarrow \infty} \frac{1}{n+1}  I(A^n;B^n)$. That is, we match the source to the channel. Therefore, we need to design the operators $\{({\cal A}_t,{\cal B}_t):~t\in\mathbb{N}^n\}$ so that the compressed signal $A_t={\cal A}_t\Gamma_t$, is sent through an AWGN channel with feedback (shown in Fig.~\ref{discrete_time_communication_system}), after which the received signal is decompressed by $\tilde{\Gamma}_t={\cal B}_tB_t$ at the pre-decoder. By the knowledge of the channel output at the decoder, the mean square estimator $\hat{X}_t$ is generated at the decoder since $\hat{X}_t\triangleq{E}\big{\{}X_t|\sigma\{\tilde{Y}^{t-1}\}\big{\}}$ (one may also use $\sigma\{{B}^{t-1}\}$ to find the filter of $\{X_t:~t\in\mathbb{N}^n\}$).\\
\noi The compression operator $\{{\cal A}_t:t\in\mathbb{N}^n\}$ is chosen so that $\lim_{n \longrightarrow \infty} R^{na}_{0,n}(D)=\lim_{n \longrightarrow \infty} \frac{1}{n+1}I(A^n;B^n)$. Recall that $B_t=A_t+Z_t$, $A_t={\cal A}_tE_tK_t$,~$Q\triangleq{C}ov(Z_t)$,~$Trace{E}\{A_tA_t^{tr}\}={P}_t, t=0,1, \ldots n$. Hence, we find $\{{\cal A}_t:t\in\mathbb{N}^n\}$ so that the following holds.
\begin{eqnarray}
\lim_{n \longrightarrow \infty} \frac{1}{n+1} C_{0,n}(P_0,\ldots{P}_n)&\triangleq& \lim_{n \longrightarrow \infty}   \frac{1}{n+1}I(A^n;{B}^n)= \lim_{n \longrightarrow \infty}\frac{1}{2}\frac{1}{n+1}\sum_{t=0}^n\log|I+E\{A_tA_t^{tr}\}Q^{-1}|\nonumber\\
&=&\lim_{n \longrightarrow \infty}\frac{1}{2}\frac{1}{n+1}\sum_{t=0}^n\log\frac{|\Lambda_{t}|}{|\Delta_{t}|}=\lim_{n \longrightarrow \infty}R^{na}_{0,n}(D).\nonumber
\end{eqnarray}
\noi From the previous equality we obtain
\begin{equation}
{\cal A}_t\triangleq\sqrt{Q\Delta_t^{-1}H_t},~t\in\mathbb{N}^n.\nonumber%\label{equation54}
\end{equation}
%Note that for $\frac{D}{p}<\min_{i}\lambda_{t,i}$, $ \forall t\in\mathbb{N}^n$, then $R_{0,n}^{na}(D)=\frac{1}{2}\frac{1}{n+1}\sum_{t=0}^n\frac{|\Lambda_t|}{(\frac{D}{p})^{p}}$.\\
The decompression operator $\{{\cal B}_t:t\in\mathbb{N}^n\}$ is chosen so that the desired distortion is achieved by the above realization. The decompressed channel output $\tilde{\Gamma}_t={\cal B}_t{B}_t$ due to transmitting the compressed input $A_t={\cal A}_t{\Gamma}_t$ is
\begin{eqnarray}
\tilde{\Gamma}_t&=&{\cal B}_tB_t={\cal B}_t(A_t+Z_t)={\cal B}_t({\cal A}_t{\Gamma}_t+Z_t),~\Gamma_t=E_tK_t\nonumber\\
&=&H_tE_tK_t+{\cal B}_tZ_t,~t\in\mathbb{N}^n.\label{equation56}
\end{eqnarray}
By pre-multiplying $\tilde{\Gamma}_t$ by $E_t^{tr}$ we can construct
\begin{eqnarray*}
\tilde{K}_t&=&E_t^{tr}\tilde{\Gamma}_t=E_t^{tr}H_tE_tK_t+E_t^{tr}{\cal B}_tZ_t,~t\in\mathbb{N}^n.
\end{eqnarray*}
The reconstruction of $Y_t$ is given by the sum of $\tilde{K}_t$ and $C\hat{X}_t$ as follows.
\begin{eqnarray}
\tilde{Y}_t&=&\Psi_t(B^t,\tilde{Y}^{t-1})\nonumber\\
&=&\tilde{K}_t+C\hat{X}_t,~\hat{X}_t=E\Big{\{}X_t|\sigma\{\tilde{Y}^{t-1}\}\Big{\}}\label{eq.13}\\
&=&E_t^{tr}H_tE_tK_t+E_t^{tr}{\cal B}_tZ_t+C\hat{X}_t,~t\in\mathbb{N}^n.\label{eq.14}
\end{eqnarray}
Next, we determine $\{{\cal B}_t:t\in\mathbb{N}^n\}$.\\
First, we notice that
\begin{eqnarray}
E\Big{\{}(Y_t-\tilde{Y}_t)^{tr}(Y_t-\tilde{Y}_t)\Big{\}}=Trace\Big{(}E\Big{\{}(Y_t-\tilde{Y}_t)(Y_t-\tilde{Y}_t)^{tr}\Big{\}}\Big{)}.\nonumber
\end{eqnarray}
Then we can compute
\begin{eqnarray}
&&E\Big{\{}(Y_t-\tilde{Y}_t)^{tr}(Y_t-\tilde{Y}_t)\Big{\}}=Trace{E}\Big{\{}(K_t-\tilde{K}_t)(K_t-\tilde{K}_t)^{tr}\Big{\}}\nonumber\\
&&=Trace{E}\Big{\{}(K_t-E_t^{tr}\tilde{\Gamma}_t)(K_t-E_t^{tr}\tilde{\Gamma}_t)^{tr}\Big{\}}\nonumber\\
&&=Trace{E}\Big{\{}(K_t-E_t^{tr}H_tE_tK_t-E_t^{tr}{\cal B}_tZ_t)(K_t-E_t^{tr}H_tE_tK_t-E_t^{tr}{\cal B}_tZ_t)^{tr}\Big{\}}\nonumber\\
&&=Trace{E}\Big{\{}\big{(}(I-E_t^{tr}H_tE_t)K_t-E_t^{tr}{\cal B}_tZ_t\big{)}\big{(}(I-E_t^{tr}H_tE_t)K_t-E_t^{tr}{\cal B}_tZ_t\big{)}^{tr}\Big{\}}\nonumber\\
&&=Trace\Big{\{}(I-E_t^{tr}H_tE_t)\Lambda_t(I-E_t^{tr}H_tE_t)^{tr}+E_t^{tr}{\cal B}_tQ{\cal B}_t^{tr}E_t\Big{\}}\nonumber\\
&&=Trace\Big{\{}(I-E_t^{tr}H_tE_t)E_t^{tr}diag(\lambda_{t,1},\ldots,\lambda_{t,p})E_t(I-E_t^{tr}H_tE_t)^{tr}+E_t^{tr}{\cal B}_tQ{\cal B}_t^{tr}E_t\Big{\}}\nonumber\\
&&=Trace\Big{\{}E_t^{tr}\Big{(}(I-H_t)diag(\lambda_{t,1},\ldots,\lambda_{t,p})(I-H_t)^{tr}+({\cal B}_tQ{\cal B}_t^{tr})\Big{)}E_t\Big{\}}\nonumber\\
&&\stackrel{(b)}=Trace\Big{\{}diag(\delta_{t,1},\ldots,\delta_{t,p})\Big{\}}=D~~~~\nonumber%\label{equation57}
\end{eqnarray}
where $(b)$ holds if we set
\begin{equation}
{\cal B}_t\triangleq\sqrt{H_t\Delta_tQ^{-1}},~t\in\mathbb{N}^n.\nonumber%\label{equation54}
\end{equation}
This shows that the realization of Fig.~\ref{discrete_time_communication_system} achieves end-to-end average distortion equal to $D$.\\
\noi{\it Decoder.} The decoder is $\tilde{Y}_t=\tilde{K}_t+C\hat{X}_t$, where $\{\hat{X}_t:~t\in\mathbb{N}^n\}$ is obtained from the modified Kalman filter as follows.
Recall that
\begin{eqnarray}
\tilde{Y}_t&=&\tilde{K}_t+C\hat{X}_t\nonumber\\
&=&E_t^{tr}H_tE_t(Y_t-C\hat{X}_t)+E_t^{tr}{\cal B}_tZ_t+C\hat{X}_t\nonumber\\
&=&E_t^{tr}H_tE_t(CX_t+NV_t-C\hat{X}_t)+E_t^{tr}{\cal B}_tZ_t+C\hat{X}_t\nonumber\\
&=&E_t^{tr}H_tE_t(CX_t-C\hat{X}_t)+C\hat{X}_t+E_t^{tr}H_tE_tN V_t+E_t^{tr}{\cal B}_t Z_t\label{equation58}
\end{eqnarray}
where $\{V_t:~t\in\mathbb{N}^n\}$ and $\{Z_t:~t\in\mathbb{N}^n\}$ are independent Gaussian vectors. Then $\hat{X}_t=E\big{\{}X_t|\sigma\{\tilde{Y}^{t-1}\}\big{\}}$ is given by the modified Kalman filter
\begin{eqnarray}
\hat{X}_{t+1}&=&A\hat{X}_t+A\Sigma_t(E_t^{tr}H_tE_tC)^{tr}M_t^{-1}\big(\tilde{Y}_t-C\hat{X}_t\big),~\hat{X}_0=\bar{x}_0\label{10}\\
\Sigma_{t+1}&=&A\Sigma_tA^{tr}-A\Sigma_t(E_t^{tr}H_tE_tC)^{tr}M_t^{-1}(E_t^{tr}H_tE_tC)\Sigma_tA+BB_t^{tr},~\Sigma_0=\bar{\Sigma}_0\label{11}
\end{eqnarray}
where
\begin{eqnarray*}
M_t=E_t^{tr}H_tE_tC\Sigma_t(E_t^{tr}H_tE_tC)^{tr}+E_t^{tr}H_tE_tNN^{tr}(E_t^{tr}H_tE_t)^{tr}+E_t^{tr}{\cal B}_t {\cal B}_t^{tr}E_t.
\end{eqnarray*}
\noi{\it Stationary Solution: Infinite Horizon.}  Now, we are ready to give the complete solution to the stationary nonanticipative RDF and its realization. As $t\longrightarrow\infty$, under the assumption that the linear Gauss-Markov system is stabilizable and detectable, we have
\begin{eqnarray}
\Sigma_{\infty}&=&A\Sigma_\infty{A}^{tr}-A\Sigma_{\infty}(E_\infty^{tr}H_\infty{E}_{\infty}C)^{tr}M_{\infty}^{-1}(E_{\infty}^{tr}H_{\infty}E_{\infty}C)\Sigma_{\infty}A+BB_{\infty}^{tr}\nonumber
\end{eqnarray}
where
\begin{equation}
M_\infty=E_\infty^{tr}H_\infty{E}_{\infty}C\Sigma_{\infty}(E_{\infty}^{tr}H_{\infty}E_{\infty}C)^{tr}+E_{\infty}^{tr}H_{\infty}E_{\infty}NN^{tr}(E_{\infty}^{tr}H_{\infty}E_{\infty})^{tr}+E_{\infty}^{tr}{\cal B}_{\infty} {\cal B}_{\infty}^{tr}E_\infty\nonumber
\end{equation}
and $E_{\infty}$ is the unitary matrix that diagonalizes $\Lambda_{\infty}$ given by
\begin{eqnarray}
E_{\infty}\Lambda_{\infty}E_{\infty}^{tr}=diag(\lambda_{\infty,1},\ldots,\lambda_{\infty,p})\nonumber
\end{eqnarray}
and
\begin{eqnarray}
\delta_{\infty,i} \triangleq \left\{ \begin{array}{ll} \xi_\infty & \mbox{if} \quad \xi_\infty\leq\lambda_{\infty,i} \\
\lambda_{\infty,i} &  \mbox{if}\quad\xi_\infty>\lambda_{\infty,i} \end{array} \right.,~i=1,\ldots,p\nonumber
\end{eqnarray}
satisfying $\sum_{i=1}^p\delta_{\infty,i}=D$.\\
Define
\begin{eqnarray}
\Delta_{\infty}=diag(\delta_{\infty,1},\ldots,\delta_{\infty,p}),~H_{\infty}=diag(\eta_{\infty,1},\ldots,\eta_{\infty,p})\nonumber
\end{eqnarray}
where $\eta_{\infty,i}=1-\frac{\delta_{\infty,i}}{\lambda_{\infty,i}}$.
The nonanticipative RDF can be computed as follows.
\begin{eqnarray}
R^{na}(D)&=&\lim_{n\longrightarrow\infty}\inf_{P_{\tilde{Y}^n|Y^n}(d\tilde{y}^n|{y}^n)\in\overrightarrow{Q}_{0,n}(D)}\frac{1}{n+1}\mathbb{I}(P_{Y^n},\overrightarrow{P}_{\tilde{Y}^n|Y^n})\nonumber\\
&=&\lim_{n\longrightarrow\infty}\Bigg{(}\frac{1}{2}\frac{1}{n+1}\sum_{t=0}^n\sum_{i=1}^p\log\Big(\frac{\lambda_{t,i}}{\delta_{t,i}}\Big)\Bigg{)}\nonumber\\
&=&\frac{1}{2}\sum_{i=1}^p\log\Big(\frac{\lambda_{\infty,i}}{\delta_{\infty,i}}\Big)=\frac{1}{2}\log\frac{|\Lambda_{\infty}|}{|\Delta_{\infty}|}=\frac{1}{2}\lim_{n\longrightarrow\infty}C_{0,n}(P_0,\ldots,P_n)\stackrel{(c)}\equiv{C}(P)\label{equation59}
\end{eqnarray}
where $(c)$ comes from the fact that the power constraint satisfies $\lim_{t\longrightarrow\infty}Trace{E}\{A_tA_t^{tr}\}=\lim_{t\longrightarrow\infty}P_t=P$. 
Thus, for a given distortion level $D$, $C(P)=R^{na}(D)$ is the minimum capacity under which there exists a realizable filter for the data reconstruction of $\{Y_t:~t\in\mathbb{N}\}$ by $\{\tilde{Y}_t:~t\in\mathbb{N}\}$ ensuring an average distortion equal to $D$. Note that for $\frac{D}{p}<\min_{i}\lambda_{\infty,i}$ then $R^{na}(D)=\frac{1}{2}\log\frac{|\Lambda_{\infty}|}{(\frac{D}{p})^p}$, e.g., $\delta_{\infty,i}=\frac{D}{p}$. Hence, from (\ref{equation59}) we have $D=p\Big(|\Lambda_{\infty}|e^{-2R^{na}}\Big)^{\frac{1}{p}}$. As a result, we have the direct relation between the reconstruction error $D$ and the rate $R^{na}$. Finally, the filter is the steady state version of (\ref{10}), (\ref{11}) with initial condition $\hat{X}_0 = E\{ X_0| Y^{-1}\}$ and $\Sigma_0$ the covariance of $X_0- \hat{X}$ which is Gaussian $N(0, \Sigma_\infty)$.

\section{Conclusion}
This paper investigates nonanticipative RDF on abstract spaces. Existence of the optimal reconstruction conditional distribution is shown, while closed form expression is derived for the stationary case. The relation between filtering theory and nonanticipative rate distortion theory is discussed via a realization procedure. Finally, an example is presented which illustrates the realization of the nonanticipative RDF.

\appendix

\section{ Proofs }
\subsection{Proof of Lemma~\ref{weakstar-closed}}
\noi To show closedness of ${\overrightarrow Q}_{ad}$ as a subset of $Q_{ad}$ it suffices to show that
\begin{eqnarray*}
\otimes_{i=0}^n{q}_i^{\alpha}(\cdot;y^{i-1},x^{i})\buildrel w^* \over \longrightarrow\otimes_{i=0}^n{q}_i^{0}(\cdot;y^{i-1},x^{i}).
\end{eqnarray*}
This will be shown by induction. Consider $n=0$. For any $h_{0}(x_0,y_0)\in{L}_1(\mu_{0},{BC}({\cal Y}_{0}))$, by definition of weak$^*$-convergence it follows from (a) that
\begin{eqnarray}
\lim_{\alpha\longrightarrow\infty}\int_{{\cal X}_0\times{\cal Y}_0}h_0(x_0,y_0)q_0^{\alpha}(dy_0;x_0)\mu_0(dx_0)=\int_{{\cal X}_0\times{\cal Y}_0}h_0(x_0,y_0)q_0^{0}(dy_0;x_0)\mu_0(dx_0)\nonumber.
\end{eqnarray}
Consider $n=1$. For $\tilde{h}_0(\cdot,\cdot)\in{L}_1(\mu_0,BC({\cal Y}_0))$, $\tilde{h}_1(\cdot,\cdot)\in{L}_1(\mu_{1},BC({\cal Y}_{1}))$ We need to show that
\begin{eqnarray}
&&\lim_{\alpha\longrightarrow\infty}\Bigg{|}\int_{{\cal X}_{0}}\bigg{(}\int_{{\cal Y}_0}h_{0}(x_0,y_0)\bigg{(}\int_{{\cal X}_{1}}\bigg{(}\int_{{\cal Y}_{1}}h_{1}(x_1,y_1)q_1^{\alpha}(dy_1;y_0,x^1)\bigg{)}\mu_1(dx_1;x_0)\bigg{)}q_0^{\alpha}(dy_0;x_0)\bigg{)}\mu_0(dx_0)\nonumber\\
&&-\int_{{\cal X}_{0}}\bigg{(}\int_{{\cal Y}_0}h_{0}(x_0,y_0)\bigg{(}\int_{{\cal X}_{1}}\bigg{(}\int_{{\cal Y}_{1}}h_{1}(x_1,y_1)q_1^{0}(dy_1;y_0,x^1)\bigg{)}\mu_1(dx_1;x_0)\bigg{)}q_0^{0}(dy_0;x_0)\bigg{)}\mu_0(dx_0)\Bigg{|}=0\nonumber.
\end{eqnarray}
The latter equation is written as follows.
\begin{eqnarray*}
&&\Bigg{|}\int_{{\cal X}_{0}}\bigg{(}\int_{{\cal Y}_0}h_{0}(x_0,y_0)\bigg{(}\int_{{\cal X}_{1}}\bigg{(}\int_{{\cal Y}_{1}}h_{1}(x_1,y_1)q_1^{\alpha}(dy_1;y_0,x^1)\bigg{)}\mu_1(dx_1;x_0)\bigg{)}q_0^{\alpha}(dy_0;x_0)\bigg{)}\mu_0(dx_0)\nonumber\\
&&-\int_{{\cal X}_{0}}\bigg{(}\int_{{\cal Y}_0}h_{0}(x_0,y_0)\bigg{(}\int_{{\cal X}_{1}}\bigg{(}\int_{{\cal Y}_{1}}h_{1}(x_1,y_1)q_1^{0}(dy_1;y_0,x^1)\bigg{)}\mu_1(dx_1;x_0)\bigg{)}q_0^{0}(dy_0;x_0)\bigg{)}\mu_0(dx_0)\Bigg{|}
\end{eqnarray*}
\begin{eqnarray*}
&&\leq\Bigg{|}\int_{{\cal X}_{0}}\bigg{(}\int_{{\cal Y}_0}h_{0}(x_0,y_0)\underbrace{\bigg{(}\int_{{\cal X}_{1}}\bigg{(}\int_{{\cal Y}_{1}}h_{1}(x_1,y_1)q_1^{0}(dy_1;y_0,x^1)\bigg{)}\mu_1(dx_1;x_0)\bigg{)}}_{\tilde{h}_1(x_0,y_0)}q_0^{\alpha}(dy_0;x_0)\bigg{)}\mu_0(dx_0)\\
&&-\int_{{\cal X}_{0}}\bigg{(}\int_{{\cal Y}_0}h_{0}(x_0,y_0)\bigg{(}\int_{{\cal X}_{1}}\bigg{(}\int_{{\cal Y}_{1}}h_{1}(x_1,y_1)q_1^{0}(dy_1;y_0,x^1)\bigg{)}\mu_1(dx_1;x_0)\bigg{)}q_0^{0}(dy_0;x_0)\bigg{)}\mu_0(dx_0)\Bigg{|}\\
&&+\Bigg{|}\int_{{\cal X}_{0}}\bigg{(}\int_{{\cal Y}_0}h_{0}(x_0,y_0)\bigg{(}\int_{{\cal X}_{1}}\bigg{(}\int_{{\cal Y}_{1}}h_{1}(x_1,y_1)q_1^{\alpha}(dy_1;y_0,x^1)\bigg{)}\mu_1(dx_1;x_0)\bigg{)}q_0^{\alpha}dy_0;x_0)\bigg{)}\mu_0(dx_0)\nonumber\\
&&-\int_{{\cal X}_{0}}\bigg{(}\int_{{\cal Y}_0}h_{0}(x_0,y_0)\bigg{(}\int_{{\cal X}_{1}}\bigg{(}\int_{{\cal Y}_{1}}h_{1}(x_1,y_1)q_1^{0}(dy_1;y_0,x^1)\bigg{)}\mu_1(dx_1;x_0)\bigg{)}q_0^{\alpha}(dy_0;x_0)\bigg{)}\mu_0(dx_0)\Bigg{|}.
\end{eqnarray*}
We need to show that both RHS terms go to zero as $a\longrightarrow\infty$. Let $\epsilon>0$ be given. Then, there exists an $\alpha_\epsilon\in{\cal D}$ such that for all $\alpha\succeq{\alpha}_\epsilon$ the first RHS term can be written as
\begin{eqnarray}
&&\Bigg{|}\int_{{\cal X}_{0}}\bigg{(}\int_{{\cal Y}_0}h_{0}(x_0,y_0){\tilde{h}_1(x_0,y_0)}\Big{(}q_0^{\alpha}(dy_0;x_0)-q_0^{0}(dy_0;x_0)\Big{)}\bigg{)}\mu_0(dx_0)\Bigg{|}\nonumber\\
&&=\Bigg{|}\int_{{\cal X}_{0}}\bigg{(}\int_{{\cal Y}_0}h_{0}(x_0,y_0)\tilde{h}_1(x_0,y_0)\Big{(}q_0^{\alpha}(dy_0;x_0)-q_0^{0}(dy_0;x_0)\Big{)}\bigg{)}\mu_0(dx_0)\Bigg{|}\nonumber\\
&&\leq\int_{{\cal X}_{0}}\Bigg{|}\int_{{\cal Y}_0}h_{0}(x_0,y_0)\tilde{h}_1(x_0,y_0)\Big{(}q_0^{\alpha}(dy_0;x_0)-q_0^{0}(dy_0;x_0)\Big{)}\Bigg{|}\mu_0(dx_0)\nonumber\\
&&\leq\epsilon,~\forall~\epsilon>0~\mbox{and}~\forall\alpha\succ\alpha_\epsilon\nonumber
\end{eqnarray}
where the last inequality follows from condition (b), e.g., $\tilde{h}_0(\cdot,\cdot)\in{L}_1(\mu_0,BC({\cal Y}_0))$.\\
The second RHS term can be written as
\begin{eqnarray}
&&\Bigg{|}\int_{{\cal X}_{0}}\bigg{(}\int_{{\cal Y}_0}h_{0}(x_0,y_0)\underbrace{\bigg{(}\int_{{\cal X}_{1}}\bigg{(}\int_{{\cal Y}_{1}}h_{1}(x_1,y_1)\Big{(}q_1^{\alpha}(dy_1;y_0,x^1)-q_1^0(dy_i;y_{0},x^1)\Big{)}\bigg{)}\mu_1(dx_1;x_0)\bigg{)}}_{\tilde{h}^\alpha_1(x_0,y_0)}
\nonumber\\
&&\otimes{q}_0^{\alpha}(dy_0;x_0)\bigg{)}\mu_0(dx_0)\Bigg{|}=\int_{{\cal X}_{0}}\int_{{\cal Y}_0}h_{0}(x_0,y_0)\tilde{h}^\alpha_1(x_0,y_0)
q_0^{\alpha}(dy_0;x_0)\otimes\mu_0(dx_0).\label{eq.15}
\end{eqnarray}
By condition (c) for $i=1$, and $\forall~\epsilon>0$ and $\alpha\succ\alpha_{\epsilon}$ we have
\begin{eqnarray}
\sup_{y_0\in{\cal Y}_{0}}\int_{{\cal X}_1}\bigg{|}\int_{{\cal Y}_1}h_1(x_1,y_1)q_1^{\alpha}(dy_1;y_0,x^1)-\int_{{\cal Y}_1}h_1(x_1,y_1)q_1^0(dy_1;y^0,x^1)\bigg{|}\mu_1(dx_1;x_0)\leq\epsilon,~\forall~x_0\in{\cal X}_0.\nonumber
\end{eqnarray}
Utilizing the last inequality into (\ref{eq.15}) yields that in the limit as $\alpha\longrightarrow\infty$, then (\ref{eq.15}) goes to zero.\\
Next, suppose that for $n=k$ and for all $\epsilon>0$ there exists $\alpha_\epsilon\in{\cal D}$ such that for any $\alpha\succeq\alpha_\epsilon$
\begin{eqnarray}
&&\Bigg{|}\int_{{\cal X}_{0}}\bigg{(}\int_{{\cal Y}_0}h_{0}(x_0,y_0)\ldots\bigg{(}\int_{{\cal X}_{k}}\bigg{(}\int_{{\cal Y}_{k}}h_{k}(x_k,y_k){q}_k^{\alpha}(dy_k;y^{k-1},x^k)\bigg{)}\mu_k(dx_k;x^{k-1})\bigg{)}\nonumber\\
&&\ldots{q}_0^{\alpha}(dy_0;x_0)\bigg{)}\mu_0(dx_0)\nonumber\\
&&-\int_{{\cal X}_{0}}\bigg{(}\int_{{\cal Y}_0}h_{0}(x_0,y_0)\ldots\bigg{(}\int_{{\cal X}_{k}}\bigg{(}\int_{{\cal Y}_{k}}h_{k}(x_k,y_k){q}_k^{0}(dy_k;y^{k-1},x^k)\bigg{)}\mu_k(dx_k;x^{k-1})\bigg{)}\nonumber\\
&&\ldots{q}_0^{0}(dy_0;x_0)\bigg{)}\mu_0(dx_0)\Bigg{|}\leq\epsilon.\nonumber
\end{eqnarray}
To conclude the derivation we need to show that for $n=k+1$
\begin{eqnarray*}
\otimes_{i=0}^{k+1}{q}_i^{\alpha}(\cdot;y^{i-1},x^{i})\buildrel w^* \over \longrightarrow\otimes_{i=0}^{k+1}{q}_i^{0}(\cdot;y^{i-1},x^{i}).
\end{eqnarray*}
Consider $n=k+1$. We need to show that for all $\epsilon>0$ there exists $\alpha_\epsilon\in{\cal D}$ such that for any $\alpha\succeq\alpha_\epsilon$
\begin{eqnarray*}
&&\Bigg{|}\int_{{\cal X}_{0}}\bigg{(}\int_{{\cal Y}_0}h_{0}(x_0,y_0)\ldots\bigg{(}\int_{{\cal X}_{k+1}}\bigg{(}\int_{{\cal Y}_{k+1}}h_{k+1}(x_{k+1},y_{k+1}){q}_{k+1}^{\alpha}(dy_{k+1};y^{k},x^{k+1})\bigg{)}\mu_{k+1}(dx_{k+1};x^{k})\bigg{)}\\
&&\ldots{q}_0^{\alpha}(dy_0;x_0)\bigg{)}\mu_0(dx_0)\nonumber\\
&&-\int_{{\cal X}_{0}}\bigg{(}\int_{{\cal Y}_0}h_{0}(x_0,y_0)\ldots\bigg{(}\int_{{\cal X}_{k+1}}\bigg{(}\int_{{\cal Y}_{k+1}}h_{k+1}(x_{k+1},y_{k+1}){q}_{k+1}^{0}(dy_{k+1};y^{k},x^{k+1})\bigg{)}\mu_{k+1}(dx_{k+1};x^{k})\bigg{)}\\
&&\ldots{q}_0^{0}(dy_0;x_0)\bigg{)}\mu_0(dx_0)\Bigg{|}\leq\epsilon.\nonumber
\end{eqnarray*}
Since,
\begin{eqnarray*}
&&\Bigg{|}\int_{{\cal X}_{0}}\bigg{(}\int_{{\cal Y}_0}h_{0}(x_0,y_0)\ldots\bigg{(}\int_{{\cal X}_{k+1}}\bigg{(}\int_{{\cal Y}_{k+1}}h_{k+1}(x_{k+1},y_{k+1}){q}_{k+1}^{\alpha}(dy_{k+1};y^{k},x^{k+1})\bigg{)}\mu_{k+1}(dx_{k+1};x^{k})\bigg{)}\\
&&\ldots{q}_0^{\alpha}(dy_0;x_0)\bigg{)}\mu_0(dx_0)\nonumber\\
&&-\int_{{\cal X}_{0}}\bigg{(}\int_{{\cal Y}_0}h_{0}(x_0,y_0)\ldots\bigg{(}\int_{{\cal X}_{k+1}}\bigg{(}\int_{{\cal Y}_{k+1}}h_{k+1}(x_{k+1},y_{k+1}){q}_{k+1}^{0}(dy_{k+1};y^{k},x^{k+1})\bigg{)}\mu_{k+1}(dx_{k+1};x^{k})\bigg{)}\nonumber\\
&&\ldots{q}_0^{0}(dy_0;x_0)\bigg{)}\mu_0(dx_0)\Bigg{|}
\end{eqnarray*}
\begin{eqnarray*}
&&\leq\Bigg{|}\int_{{\cal X}_{0,k}}\int_{{\cal Y}_{0,k}}\otimes_{i=0}^kh_{i}(x_i,y_i)\bigg{(}\int_{{\cal X}_{k+1}}\int_{{\cal Y}_{k+1}}h_{k+1}(x_{k+1},y_{k+1})\Big{(}{q}_{k+1}^{\alpha}(dy_{k+1};y^{k},x^{k+1})\nonumber\\
&&-q_{k+1}^0(dy_{k+1};y^{k},x^{k+1})\Big{)}\mu_{k+1}(dx_{k+1};x^{k})\bigg{)}
\otimes_{i=0}^k{q}_{i}^\alpha(dy_i;y^{i-1},x^i)\otimes\mu_{i}(dx_i;x^{i-1})\Bigg{|}\\
&&+\Bigg{|}\int_{{\cal X}_{0,k}}\int_{{\cal Y}_{0,k}}\otimes_{i=0}^kh_{i}(x_i,y_i)\underbrace{\bigg{(}\int_{{\cal X}_{k+1}}\int_{{\cal Y}_{k+1}}h_{k+1}(x_{k+1},y_{k+1}){q}_{k+1}^{0}(dy_{k+1};y^{k},x^{k+1})\mu_{k+1}(dx_{k+1};x^{k})\bigg{)}}_{\tilde{h}_{k+1}(x^k,y^k)}\\
&&\otimes_{i=0}^k\Big{(}{q}_{i}^\alpha(dy_i;y^{i-1},x^i)-{q}_{i}^0(dy_i;y^{i-1},x^i)\Big{)}\otimes\mu_{i}(dx_i;x^{i-1})\Bigg{|}.\nonumber
\end{eqnarray*}
By condition (c) the following inequality holds, $\forall x^k\in{\cal X}_{0,k}$,
\begin{align*}
&\sup_{y^k\in{\cal Y}_{0,k}}\int_{{\cal X}_{k+1}}\bigg{|}\int_{{\cal Y}_{k+1}}h_{k+1}(x_{k+1},y_{k+1})\Big{(}{q}_{k+1}^{\alpha}(dy_{k+1};y^{k},x^{k+1})-\\
&\qquad\qquad{q}_{k+1}^0(dy_{k+1};y^{k},x^{k+1})\Big{)}\bigg{|}\mu_{k+1}(dx_{k+1};x^{k})\nonumber\\
&\leq\epsilon,~\forall~\epsilon>0~\mbox{and}~\forall~\alpha\succ\alpha_{\epsilon}.
\end{align*}
Also, by condition (b), $\tilde{h}_{k+1}\in{L}_1(\mu_{0,k},BC({\cal Y}_{0,k}))$. Utilizing the previous observations and the induction hypothesis $\otimes_{i=0}^{k}{q}_i^{\alpha}(\cdot;y^{i-1},x^{i})\buildrel w^* \over \longrightarrow\otimes_{i=0}^{k}{q}_i^{0}(\cdot;y^{i-1},x^{i})$ in the two inequalities above, then in the limit as $\alpha\longrightarrow\infty$, the terms in the inequality go to zero.\\
As a result, ${\overrightarrow Q}_{ad}$ is a weak$^*$-closed set. Being a weak$^*$-closed subset of the weak$^*$-compact set $Q_{ad}$, ${\overrightarrow{Q}}_{ad}$ is also weak$^*$-compact.

\subsection{Proof of Theorem~\ref{lagrange_duality}}
\noi The proof is based on  Lagrange Duality theorem \cite[Theorem 1, p.~224]{dluenberger69}. We choose $X \triangleq{L}_{\infty}^{w}(\mu_{0,n},M_{rba}({\cal Y}_{0,n}))$ which is clearly a vector space. For the set $\Omega$ the natural choice is the set $ \Omega  =\overrightarrow{Q}_{ad}\equiv{L}_{\infty}^w(\mu_{0,n},{\Pi}_{rba}({\cal Y}_{0,n}))\subseteq X$. Define
\begin{eqnarray}
G(\overrightarrow{q}_{0,n})&\triangleq& \ell_{d_{0,n}}({\overrightarrow q}_{0,n})-D,\quad {\overrightarrow q}_{0,n} \in L_{\infty}^{w}(\mu_{0,n},M_{rba}({\cal Y}_{0,n}))\nonumber\\
&\triangleq&  \int_{{\cal X}_{0,n}}
\biggl( \int_{{\cal Y}_{0,n}}d_{0,n}(x^n,y^n) {\overrightarrow q}(dy^n;x^n) \biggr) \mu_{0,n}(dx^n) - D.\nonumber
\end{eqnarray}
It is clear that $G(\cdot)$ is a convex mapping from $L_{\infty}^{w}(\mu_{0,n},M_{rba}({\cal Y}_{0,n}))$ into the real line with the natural ordering $(\mathbb{R},\preceq) \triangleq Z$. Also recall that ${\overrightarrow q}_{0,n} \rightarrow {\mathbb{I}}(\mu_{0,n}; \overrightarrow{q}_{0,n})$
is convex  and well defined on $\Omega$  and that,  by Theorem \ref{th3}, $\inf_{\overrightarrow{q}_{0,n}\in \overrightarrow{Q}_{0,n}(D)}\mathbb{I}(\mu_{0,n}; \overrightarrow{q}_{0,n})$ exists and is finite. Thus, according to the Lagrange duality theorem referred to above, it suffices to show that there exists a ${\overrightarrow q}_{0,n}^{1} \in \Omega$ such that
\begin{eqnarray}
G({\overrightarrow q}_{0,n}^{1})= \int_{{\cal X}_{0,n}} \big\{ \int_{{\cal Y}_{0,n}}d_{0,n}(x^n,y^n){\overrightarrow q}_{0,n}^{1}(dy^n;x^n)\big\} \mu_{0,n}(dx^n) - D < 0.\nonumber
\end{eqnarray}
Introduce  the sets $A_1 \triangleq \{ x^n \in {\cal X}_{0,n}: \Gamma_{x^n} \ne \emptyset \} $ and $A_0 \triangleq {\cal X}_{0,n}\setminus A_1$, with $\Gamma_{x^n}$ denoting the $x^n$-section of $\Gamma.$  Define the measure valued function ${\overrightarrow q}_{0,n}^{1}$ as follows
\begin{eqnarray}
{\overrightarrow q}_{0,n}^{1}(\Gamma_{x^n};x^n) = 0,~~ \forall ~~ x \in A_0;~~ {\overrightarrow q}_{0,n}^{1}({\cal Y}_{0,n};x^n)= 1,~~ \forall~~ x^n \in {\cal X}_{0,n}\nonumber
\end{eqnarray}
\begin{eqnarray}
0 \leq {\overrightarrow q}_{0,n}^{1}(B;x^n)\leq 1, B\subset \Gamma_{x^n}, ~ {\overrightarrow q}_{0,n}^{1}(\Gamma_{x^n};x^n) =1,~~ \forall ~~ x^n \in A_1\nonumber
\end{eqnarray}
where $B \in {\cal B}({\cal Y}_{0,n})$. Since by hypothesis $\Gamma \ne \emptyset$ we have  $\mu_{0,n}(A_1)>0$
and thus the kernel ${\overrightarrow q}_{0,n}^{1}$ is well defined and it belongs to $L_{\infty}^w(\mu_{0,n},{\Pi}_{rba}({\cal Y}_{0,n}))$. Using this kernel in  the expression for $\ell_{d_{0,n}}({\overrightarrow q}_{0,n})$, one can easily verify that $\ell_{d_{0,n}}({\overrightarrow q}_{0,n}^{1}) < D$ and hence $G( {\overrightarrow q}_{0,n}^{1}) < 0$. Then, by the Lagrange Duality theory, we arrive at the conclusion of the theorem as stated. Also it follows from the same duality theory that if the infimum is achieved by some ${\overrightarrow q}_{0,n}^* \in L_{\infty}^w(\mu_{0,n},{\Pi}_{rba}({\cal Y}_{0,n}))$, then
\begin{eqnarray}
s \biggl(\int_{{\cal X}_{0,n}} \int_{{\cal Y}_{0,n}}d_{0,n}(x^n,y^n) {\overrightarrow q}_{0,n}^*(dy^n;x^n)\otimes\mu_{0,n}(dx^n) - D\biggr
)=0.
\end{eqnarray}
In other words, for non-zero $s \in (-\infty,0]$, solution occurs on the boundary. This completes the proof.

\section*{Acknowledgement}
The authors wish to thank the associate editor and anonymous reviewers for their valuable comments which improved significantly the presentation of this paper.

\bibliographystyle{IEEEtran}
\bibliography{photis_references_filtering_weakstar}

\end{document}